\documentclass[11pt]{article}
\usepackage{latexsym}
\usepackage{graphicx}
	\graphicspath{{./}{figures}}
\usepackage{xcolor}
\usepackage{epstopdf}
\usepackage{amssymb}
\usepackage{amsthm}
\usepackage{amsmath,enumerate,rotate}
\usepackage{algorithm}
\usepackage{algpseudocode}
\usepackage{hyperref}
\usepackage[super]{nth}
\newcommand{\R}{\mathbb R}

\def\be#1\ee{\begin{equation}#1\end{equation}}
\newcommand{\fer}[1]{(\ref{#1})}

\newtheorem{proposition}{\bf Proposition}[section]

\newtheorem{remark}{\bf Remark}[section]

\newcommand{\bq}{\begin{equation}}
\newcommand{\eq}{\end{equation}}

\def\bqa{\begin{eqnarray}}
\def\eqa{\end{eqnarray}}

\def\e{\epsilon}

\newcommand{\bd}{\begin{displaymath}}
\newcommand{\ed}{\end{displaymath}}
\newcommand{\ba}{\begin{eqnarray}}
\newcommand{\ea}{\end{eqnarray}}

\def\R{\mathbb{R}}

\newenvironment{equations}{\equation\aligned}{\endaligned\endequation}

\begin{document}

\title{A data-driven kinetic model for opinion dynamics\\ with social network contacts}

\author{G. Albi\thanks{Department of Computer Science, University of Verona, Strada Le Grazie 15, 37134 Verona, Italy. (giacomo.albi@univr.it)} \and E. Calzola \thanks{Department of Computer Science, University of Verona, Strada Le Grazie 15, 37134 Verona, Italy. (elisa.calzola@univr.it)}  \and
	G. Dimarco\thanks{Department of Mathematics and Computer Science \& Center for Modeling, Computing and Statistics (CMCS), University of Ferrara, via Machiavelli 30, 44121 Ferrara, Italy. (giacomo.dimarco@unife.it)}} 

\maketitle

\begin{abstract}
	Opinion dynamics is an important and very active area of research that delves into the complex processes through which individuals form and modify their opinions within a social context. The ability to comprehend and unravel the mechanisms that drive opinion formation is of great significance for predicting a wide range of social phenomena such as political polarization, the diffusion of misinformation, the formation of public consensus, and the emergence of collective behaviors. In this paper, we aim to contribute to that field by introducing a novel mathematical model that specifically accounts for the influence of social media networks on opinion dynamics. With the rise of platforms such as Twitter, Facebook, and Instagram and many others, social networks have become significant arenas where opinions are shared, discussed, and potentially altered. To this aim after an analytical construction of our new model and through incorporation of real-life data from Twitter,  we calibrate the model parameters to accurately reflect the dynamics that unfold in social media, showing in particular the role played by the so-called influencers in driving individual opinions towards predetermined directions. 
\end{abstract}


\tableofcontents

\section{Introduction}
In recent years, kinetic models and specifically Boltzmann equations have emerged as very powerful tools for describing and analyzing the collective behaviors exhibited by systems of interacting agents \cite{PT13,Bellomo}. These models have found applications across a diverse range of fields giving a contribution to the advancement of the knowledge in various disciplines. For example in economics, kinetic models have been recently used to study the dynamics of market prices and trading outcomes (\cite{Cor05,Par14,Cha07,T22,Burger,DPT18}). Using the tools of kinetic theory in financial markets to study the evolution of prices, it is possible to predict the emergence of bubbles and crashes \cite{Li22}. Moreover, Boltzmann type equations permit to characterize the details of the economical interactions and to describe the wealth distribution in a society and the appearance of inequalities \cite{Cor05,Par14}. 
In biology, kinetic models have been extensively used to study the dynamics of populations and the spread of epidemics (see, e.g., \cite{Dim21,Cha04,ZBDDPAFT,BDP21,DTZ221}). The kinetic theory of infectious diseases has been shown to be a powerful mean to describe the spread of a disease in a homogeneous population and also when spatial differences become a key aspect to accurately describe a pandemic \cite{BDP21,BBDP21}. These models can be used to predict the effectiveness of vaccination and quarantine measures in a population \cite{DTZ221} and can be efficiently interfaced with data \cite{ZBDDPAFT}. Boltzmann-type equations have also been used to study the evolution of cooperation and altruism in social systems \cite{albi2017opinion,Burger1} and to analyze the dynamics of genetic mutations (\cite{Tos13}). The application of kinetic models in social sciences has also a long and successful history \cite{Tos14,Gal82,DT19,Dim20}, in this context, we also recall the recent study of information diffusion \cite{fakenews} which can be pursued through these mathematical techniques. In engineering, kinetic models have been used to study the dynamics of traffic flow \cite{DT201,DTZ222} notably the formation and dissipation of traffic jams and the prediction of the effects of traffic control measures (see \cite{Agn15,Gun02,Pup16}). Crowd behaviors \cite{Agn15,ABFH19,BGQR22,Wolfram} and network communication with emphasis on the optimization of communication protocols in wireless networks \cite{Dur09,To06} can be studied as well.

In the above depicted and wide framework of application of kinetic theory to social and biological systems, opinion formation, i.e. the dynamics of how opinions evolve and spread among individuals plays a relevant role due to its importance in the society: political polarization and consensus among others. To shed light on this complex process, the methods of statistical physics have proven to be highly effective and efficient tools for studying and analyzing such phenomena \cite{To06,Dur09,APZ14,APTZ17}. In this regard, one of the key concepts from statistical physics that can be applied to the dynamics of opinion formation is the notion of emergent behavior. Emergence refers to the phenomenon where collective properties and behaviors arise from the interactions and dynamics of individual components. In the context of opinion formation, emergent behavior can manifest as opinion clusters, polarization, consensus, or the formation of influential opinion leaders.

In this paper, we introduce a new kinetic model through the use of probabilistic and statistical tools describing the microscopic dynamics of opinion formation and change. We then upscale our model at the bound of observable quantities and we provide a quantitative framework for analyzing social phenomena and designing interventions that promote constructive dialogue and reduce polarization. One key ingredient of our study consists in the description of the role played by individuals with strong ascendancy on the rest of the population through a detailed analysis of the network to which the population belongs and the high number of connections of such few influential individuals. More in detail, our model describes the evolution of opinions starting from the microscopic bound. Each agent/individual has associated two real values: its number of followers on a given social media platform where he/she is used to interact and its opinion. By assuming that the number of followers is not influenced by their opinion, we first construct an evolutionary model describing the connections among individuals over a fixed social platform. We successively restrict ourselves to a given social network, namely Twitter, and we show how the proposed model is well adapted to describe Twitter networks by matching real data, through a parameter estimation technique, with the equilibrium distribution obtained with this new contact model.

In the second part, we assume that opinions are continuous variables lying on a bounded interval $[-1,1]$ indicating respectively total disagreement and total agreement about a given topic, and that the agents update their opinions after the interaction with others through the social platform. The strength of such interaction is supposed to depend on the number of followers of each agent and on the distance between their opinions. We also suppose that there is a certain amount of randomness in the interaction, modeling external factors which can be hardly controlled such as the possibility to access information and the knowledge of every single individual. Under these hypotheses, we derive the kinetic equation that describes the time evolution of the distribution of opinions in presence of social media contacts in the population, and we finally study its properties using analytical and numerical methods. In the last part, we show that our kinetic model is able to capture important features of opinion dynamics, such as the emergence of consensus and polarization, according to the choice of different interaction kernels between the agents. Furthermore, to comprise data with opinion dynamics, we use a {\em sentiment analysis} (SA) method to assign a score to textual information ({\em the tweets}). We use this method to represent the agents' opinions 
extracting data from Twitter over some specific topic and assigning a score in the interval $[-1,1]$. We mention that SA, also known as {\em Opinion Mining} is a subclass of Natural Language Process (NLP) methods to analyze textual information, in this context we refer to  \cite{zhang2008text,medhat2014sentiment,zhou2013sentiment,Hutto} for further details.
Finally, to fit the actual trend of the opinion distribution of the agents, we calibrate the interaction kernels that rule the evolution of the dynamics, using a parameter estimation approach based on the minimization of a loss function, between data extracted from Twitter, and the result of the simulation. Our results provide insights into the mechanisms that drive opinion formation and changes on a social platform.

The rest of the work is structured as follows. In Section \ref{know5} we describe how to model the formation of a network on social media platforms starting from a microscopic approach and successively upscale the description at a mesoscopic bound by deriving a kinetic equation for the time evolution of the connections. We show that for different choices of a so-called value function appearing in every single interaction among individuals, we are able to recover different stationary distributions for the distribution of contacts, and we explain how to use a given dataset from Twitter in order to select the set of parameters that allows describing the current state of the connections on that specific social network. Section \ref{model} is devoted to the model of opinion formation given the presence of social media contacts. Starting from the binary interaction between two agents, which takes into account the compromise propensity of the agents and a certain amount of randomness in the process, we recover the Boltzmann-like equation that describes the evolution of the density of the joint distribution of opinions and contacts. In Section \ref{nume}, we perform some numerical simulations. In the first part, we analyze the qualitative behavior of our new model and we illustrate its capabilities in describing different artificial situations. In the second part, we focus on real data extracted from Twitter, we perform a sentiment analysis to obtain an opinion distribution and we reconstruct with our new model the interaction kernel that leads to the opinion distribution derived from that sentiment analysis. A last Section \ref{sec:conc} is devoted to drawing some conclusions and individuating some future research axes.

\section{An evolutionary model for contacts }\label{know5}
This section is dedicated to the construction of a mathematical model describing social contacts on the web with an emphasis on social platform networks. 
There exists a vast literature on network modeling, see for instance \cite{Ace11,Das14,Do15,APTZ17} and the reference therein where empirical, Bayesian, and non Bayesian methods and probability approaches based on Poisson distribution are discussed and employed. Here, the path followed is different and it sinks its roots in the interplay between the kinetic theory of gases \cite{Cer,PT13} and the prospect theory of Kahneman and Tversky \cite{KTa} which was first introduced in the context of behavioral economic studies to characterize the science of decision in a population. Our aim is to exploit this theory with the scope of building a model which is able to describe the evolution in time of the number of contacts on a social media platform through the methods and techniques of kinetic theory. Let us observe that a similar study has been performed in \cite{Dim21} for characterizing the contact dynamics related to the spread of an epidemic. However, in the context of the virtual contacts, which are the ones to take place in this study, the results are different as well as the type of equilibrium distributions characterizing the network obtained as shown later. One additional and important point to highlight is that the results achieved in this section permit us to match the real data taken from a given web platform, namely Twitter, with high precision.

In the next section, through modeling choices discussed here for the network formation, we will introduce a detailed description of the microscopic binary interactions taking place among individuals acting on a social network for what concerns the formation of opinions. We will in particular shed light on the role played by individuals with a large number of connections in driving others' opinions. In the sequel, we will often refer to individuals with a number of connections larger than the average to as the influencers. Let us also observe that the choice we will do successively of giving these individuals a larger weight in the opinion balance is consistent with the actual current functioning of most social media platforms: individuals are exposed to content created by popular users more often than the ones posted by their local connections and more likely influenced by the former. For the moment, and for the sake of clarity, we restrict ourselves to the sole case of contact dynamics and we will follow the construction in \cite{Dim20} to derive the master equation of Boltzmann type that describes the evolution of such quantity. 

We consider then a system of agents characterized by the number of their social media followers which, from now on, we will refer to as $c>0$. We assume that our agents are indistinguishable and that, at time $t \geq 0$, they are only characterized by the number of their contacts. The concept of contact or connection here has to be intended as the number of individuals following the contents spread over a given platform by a given second subject. One can then suppose the statistical distribution of contacts/connections of the agents to be fully characterized by the density $h(c,t)$ of contacts, which is such that, given the sub-domain $D\subseteq \mathbb{R}_+$, the integral
\[
\int_D h(c,t)\mathrm{d}c
\]
represents the number of people having $c\in D$ followers at time $t>0$. The density function $h$ is assumed to be normalized to one, so that
\[
\int_{\mathbb{R}_+} h(c,t)\mathrm{d}c = 1.
\]
In analogy with the problem of social climbing presented in \cite{Dim20} where agents were aiming to climb the social ladder to reach high social status, here it is reasonable to assume that the formation of a given network begins due to the will of all the participants to interact and to be heard by others. The consequence is that each individual, we usually refer to him/her to as an agent in the sequel, likely wants to increase its number of contacts when he enters the platform by interacting with other agents. We assume that, for most participants, there exists a given number of social contacts, $\bar c$, such that they considered themselves satisfied when this number is reached. Moreover, the elementary interaction which takes place at the microscopic bound will express the tendency of the agents to reach, at least, the number of followers equals to $\bar c$. 
In order to describe the evolution of social contact dynamics, we will now take inspiration from the prospect theory of Kahneman and Tversky \cite{KTa} and we introduce a value function $\Psi_\delta$ modeling this behavioral theory. We also have
\begin{equation}\label{k1}
	c' = c - \Psi_\delta(c/\bar c) c +\eta c
\end{equation}
where $0 \leq \delta \leq 1$ is a parameter characterizing the intensity of the individual behavior and $\eta$ is a random variable with zero mean and finite standard variation $\nu$. In \eqref{k1}, we are then assuming that each individual tries to increase the number of its followers by getting in touch with friends or by sharing content that may be of interest to others, voicing strong opinions or simply using their social status to capture interest. The result of such dynamics is that the number of followers of each agent can be modified for two reasons, one quantified by a deterministic value function $\Psi_\delta$, which can assume both negative and positive values, i.e. it is possible that with its actions an individual gain or lose connections. The second reason which may lead to a change in the number of followers is due to the intrinsic unpredictability of this complex process which is quantified consequently by a random variable $\eta$ with zero mean. The value function $\Psi_\delta(s)$, $s\geq 0$, encodes the properties of the theory exposed in \cite{KTa}: it is a dimensionless increasing function equal to zero at the point $s=1$ corresponding to the point where most individuals reach satisfaction in terms of their role in the network and it verifies the conditions
\begin{equation}\label{eq:p1}
	-\Psi_\delta(1 - \Delta s) > \Psi_\delta(1 + \Delta s)
\end{equation}
and
\begin{equation}\label{eq:pr2}
	\frac{\mathrm{d}}{\mathrm{d}s}\Psi_\delta(s)\big|_{1 + \Delta s}<\frac{\mathrm{d}}{\mathrm{d}s}\Psi_\delta(s)\big|_{1 - \Delta s}
\end{equation}
for $0 < \Delta s \leq 1$. Request \eqref{eq:pr2} implies that the value function is asymmetric, meaning that it is steeper below the reference point $s=1$ than above it. This models the fact that, if two agents start at the same distance $\Delta s$ from the reference point $s=1$, getting closer to $s=1$ will be easier for the agent starting from below the reference point than for the one starting above. It is quite easy for individuals with a low number of followers to increase the number of their connections by contacting friends and related, while it is more difficult to decrease the number of followers when a certain status, i.e. $c>\bar c$, is reached, and in general, an individual is not interested in decreasing this number. Setting $s=c/\bar c$ we choose the precise form of the value function which reads
\begin{equation}\label{eq:vf}
	\Psi_\delta (s) = -\mu \frac{e^{(s^{-\delta} - 1)/\delta}-1}{(1-\mu)e^{(s^{-\delta} - 1)/\delta} + 1 + \mu}.
\end{equation}
The above equation respects all the properties detailed above for any value of $0< \delta \leq 1$. It has an inflection point $\bar s<1$ and it is convex in $[0,\bar s]$ and concave for $s>\bar s$. This inflection point corresponds to a certain value $\hat c < \bar c$ below which, in principle, the agents do not expect to increase their number of followers, while the satisfactory number of followers $c=\bar c$ corresponds to the reference point $s=1$. 

	Moreover, the value function \eqref{eq:vf}, is bounded by the following relation
	\begin{equation}\label{eq:psibounds}
		-\frac{\mu}{1-\mu}\leq\Psi_\delta\leq - \mu\frac{e^{-1/\delta}-1}{(1-\mu) e^{-1/\delta}+1+\mu}
	\end{equation}
	for  $0\leq\delta\leq1$. Concerning the role played by the parameter $\delta$, one can notice that the smaller is this value, the easiest is the possibility to gain some followers when one is below the value $c=\bar c$.
We report in Figure \ref{fig:ValueFun} different shapes of the value function, for different choices of $\delta$, where the dashed red lines represent the bounds \eqref{eq:psibounds}.

\begin{figure}[h!]
	\centering
	\includegraphics[width=0.75\textwidth]{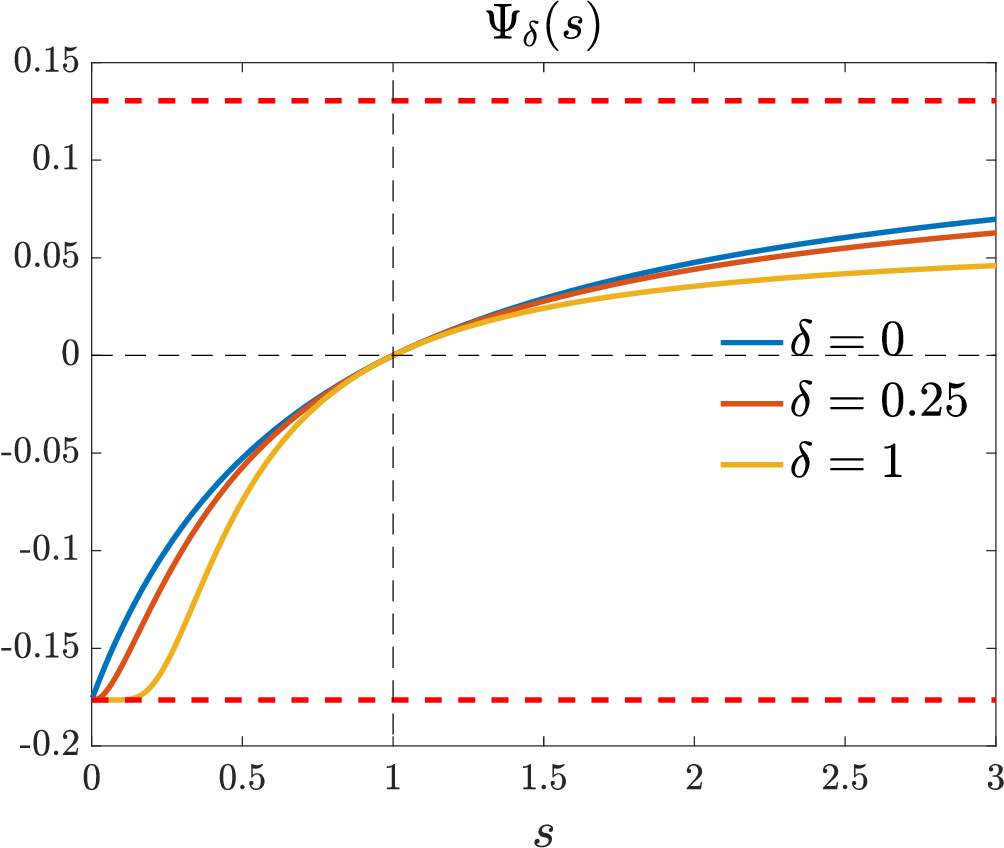} 
	\caption{Profiles of the value function \eqref{eq:vf} for different choices of $\delta$ and $\mu=0.15$. The red dashed lines represent the bounds \eqref{eq:psibounds}}
	\label{fig:ValueFun}
\end{figure}
Before concluding this part, we introduce a rescaling factor in equation \eqref{k1} meaning that we are interested in studying a process in which the formation of this social network is a consequence of small upgrading of the number of connections in time. The rescaled equation reads
\begin{equation}\label{k1eps}
	c' = c - \Psi^\e_\delta(c/\bar c) c +\eta_\e c
\end{equation}
where now
\begin{equation}\label{eq:vfeps}
	\Psi^\e_\delta (s) = -\mu \frac{e^{\e(s^{-\delta} - 1)/\delta}-1}{(1-\mu)e^{\e(s^{-\delta} - 1)/\delta} + 1 + \mu}.
\end{equation}
$\eta_{\e}$ is the same random variable as before but with variance $\e\nu^2$ and $\e$ is a small parameter. The role of $\e$ in the dynamics will be clarified in the next section.

\subsection{The kinetic model for the evolution of social media contacts}\label{kinetic_model}
We aim now in deriving an evolutionary model for the formation of a network on a social media platform resorting to the contact law \eqref{eq:vfeps} derived previously. To that aim, let us observe that the variation of the density $h_\epsilon (c,t)$, also rescaled through the parameter $\e$, obeys a linear Boltzmann-like equation \cite{PT13} whose weak form corresponds to
\begin{equation}\label{eq:ble}
	\frac{\mathrm{d}}{\mathrm{d}t}\int_{\mathbb{R}_+} \varphi (c) h_\epsilon (c,t) \mathrm{d}c = \big\langle \int_{\mathbb{R}_+} \chi (c)\left( \varphi (c_*) - \varphi (c) \right) h_\epsilon (c,t) \mathrm{d}c \big\rangle,
\end{equation}
for all smooth test functions $\varphi(c)$. These functions are the so-called observable quantities of the underlying random process. For example taking $\varphi(c)=1$ leads \eqref{eq:ble} to an equation for the time evolution of the number of individuals in the network which can be easily inferred from \eqref{eq:ble} is constant in time. Instead, the case $\varphi(c)=c$ leads to an evolution equation for the average number of connections in the network which can be inferred, contrary to the previous case, not to be conserved in time. The positive function $\chi (c)$ measures the frequency of the interactions with $c$ followers and the expectation $\langle \cdot \rangle$ takes into account the presence of the random variable $\eta_\epsilon$. More in detail $\langle \cdot \rangle$ gives the expected value with respect to the random space in which $\eta_{\e}$ lives. In order to preserve the positivity of the connections in \eqref{k1} and in the rescaled equation \eqref{k1eps} based on the bounds of the value function \eqref{eq:psibounds}, we require the random variable $\eta_\e$ to be uniformly distributed and to take values in
\begin{equation}\label{eq:boundeta}
		\eta_\e\in\left[-\left|\frac{e^{-\e/\delta}+1}{(1-\mu)e^{-\e/\delta} + 1 + \mu}\right|,\left|\frac{e^{-\e/\delta}+1}{(1-\mu)e^{-\e/\delta} + 1 + \mu}\right|\right].
\end{equation}
In the sequel, we consider collision kernels in the form \cite{Fur}
\begin{equation}\label{eq:ck}
	\chi (c) = c^\beta \alpha,
\end{equation}
for some multiplicative positive constants $\alpha >0$ and for exponents $\beta \geq 0$. In order to establish reasonable values for $\chi(c)$, one can observe that, if $s>0$, the individual rate of growth ${\partial \Psi_\delta^\epsilon (s)}/{\partial s}$ vanishes as $\epsilon \to 0$ (cfr. \cite{Dim20}). So, to maintain a collective growth different than zero for all values of the scaling parameter $\epsilon$, one suitable choice is to take 
\begin{equation} \label{freq}
	\alpha = \frac{1}{\kappa\epsilon},
\end{equation}
corresponding to a frequency of interaction proportional to $1/\epsilon$ with instead $\kappa$ an order 1 constant. Concerning the second parameter $\beta$, we will consider two different situations. The first consists in taking $\beta>0$ which implies that the frequency of interaction becomes greater if the number of connections is higher. This situation is encountered in social platforms where the type of exchanges are typically one to one and consequently one shares contents proportionally to the number of its connections.
The second case we will consider is $\beta=0$ which corresponds to the situation in which interactions are independent of the number of connections meaning that the activity of each individual is independent of the other and the content sharing is independent of the size of the relative network. 
In the case $\beta>0$ a rational choice would consist in setting $\beta=\delta$ when $\delta$ is positive. Choices of $\beta < \delta$ will imply very high variations of the collective growth of the number of social media connections when $c$ is small compared to the case with $c$ large. However, it is not reasonable to expect that individuals with few connections can reach an influential position easily. Conversely, $\beta>\delta$ implies very small variations of collective growth for small values of $c$, excluding consequently the opposite situation, i.e. the possibility to increase the number of followers for individuals having few connections. Hence, choosing $\beta = \delta$ is a good compromise between the two scenarios.  With the above discussed choices then equation \eqref{eq:ble} can be rewritten as
\begin{equation}\label{eq:ble1}
	\frac{\mathrm{d}}{\mathrm{d}t}\int_{\mathbb{R}_+} \varphi (c) h_\epsilon (c,t) \mathrm{d}c = \frac{1}{\epsilon \kappa} \big\langle \int_{\mathbb{R}_+}c^\delta\left( \varphi (c_*) - \varphi (c) \right) h_\epsilon (c,t) \mathrm{d}c \big\rangle.
\end{equation}
To get some insight into the time evolution of the model above we use a standard procedure borrowed from the theory of gases and so we expand in Taylor series $\varphi(c')$ around $\varphi(c)$ supposing $\varphi(c)$ smooth enough. We have that
\[
\langle c' -c\rangle = -\Psi_\delta^\epsilon(c/\bar c) c, \qquad \langle (c' - c)^2 \rangle = \left(\Psi_\delta^\epsilon(c/\bar c)\right)^2 c^2 + \epsilon \nu^2 c^2,
\]
One can also observe that for $0<\delta\leq 1$, it holds that 
\begin{equation}\label{rem} 
	\lim_{\epsilon \to 0} \frac{1}{\epsilon} \Psi_\delta^\epsilon \left( \frac{c}{\bar c} \right) = \frac{\mu}{2\delta} \left(1- \left( \frac{\bar c}{c}\right)^\delta \right).
\end{equation} 
This gives
\[
\langle \varphi(c') - \varphi(c) \rangle = \epsilon\left(-\varphi'(c) \frac{1}{\epsilon}\Psi_\delta^\epsilon(c/\bar c) c + \frac{\nu^2}{2} \varphi''(c) c^2\right) + R_\epsilon(c),
\]
where $ R_\epsilon(c)$ is a remainder of the Taylor expansion such that $R_\epsilon(c) = o(\epsilon)$ thanks to \eqref{rem}. Therefore, using for the interaction frequency \eqref{freq}, we get for the evolution of the observable $\varphi(c)$
\[
\frac{\mathrm{d}}{\mathrm{d}t} \int_{\mathbb{R}_+} \varphi (c) h_\epsilon (c,t) \mathrm{d}c = \int_{\mathbb{R}_+}\frac{c^\delta}{\kappa}\left( -\varphi'(c) \frac{1}{\epsilon}\Psi_\delta^\epsilon(c/\bar c) c + \frac{\nu^2}{2} \varphi''(c) c^2 \right)h_\epsilon (c,t) \mathrm{d}c + \frac{1}{\kappa\epsilon}\mathcal{R}_\epsilon(c,t)
\]
where 
\[
\mathcal{R}_\epsilon (c,t) = \int_{\mathbb{R}_+} R_\epsilon(c)h_\epsilon(c,t) \mathrm{d}c 
\]
and by using \eqref{rem} we get the following approximation
\begin{equation}\label{eq:fp}
	\frac{\mathrm{d}}{\mathrm{d}t} \int_{\mathbb{R}_+} \varphi (c) h (c,t) \mathrm{d}c = \int_{\mathbb{R}_+}\left( -\varphi'(c)\frac{\tilde \mu}{2\delta}\left(1- \left( \frac{\bar c}{c} \right)^\delta \right)c^{1+\delta} + \frac{\tilde\nu^2}{2} \varphi''(c) c^{2+\delta} \right)h(c,t) \mathrm{d}c,
\end{equation}
in which we have set $\tilde \mu = \mu/\kappa$ and $\tilde \nu^2 = \nu^2/\kappa$.
Under the additional hypothesis that the boundary terms produced by the integration by parts vanish, i.e. a zero flux condition, equation \eqref{eq:fp} is the weak form of the following Fokker-Planck equation
\begin{equation}\label{eq:fpstrong}
	\frac{\partial h(c,t)}{\partial t} = \frac{\tilde \mu}{2\delta}
	\frac{\partial}{\partial c} \left( \left(1- \left( \frac{\bar c}{c} \right)^\delta \right)c^{1+\delta}h(c,t)\right) + \frac{\tilde \nu^2}{2} \frac{\partial^2}{\partial c^{2}}\left(c^{2+\delta}h(c,t)\right),
\end{equation}
that describes the evolution of the density of contacts $c\in \mathbb{R}_+$ in the limit of the quasi-invariant variations of followers.

We are now interested in a steady state solution of equation \eqref{eq:fpstrong}. In fact, for the time of dynamics which we aim to study, i.e. the one related to the formation of opinions through social interactions on online platforms, one can reasonably suppose that the connectivity network is stationary being the time at which opinions about a given subject are shaped much faster than the changes in the network. Thus, one can observe that the equilibrium solution of \eqref{eq:fpstrong} is a function solving the first order differential equation
\begin{equation}\label{eq:stat_state} \frac{\tilde \nu^2}{2} \frac{\mathrm{d}}{\mathrm{d}c}\left(c^{2+\delta}h(c)\right) +  \frac{\tilde \mu}{2\delta} \left(1- \left( \frac{\bar c}{c} \right)^\delta \right)c^{1+\delta}h(c)= 0
\end{equation}
To find a solution to \eqref{eq:stat_state}, we perform the change of variable $\rho(c) = c^{2+\delta} h(c)$ and by setting $\gamma = \tilde \mu/\tilde \nu^2 = \mu / \nu^2$, one can easily observe that the function $\rho(c)$ solves the following equation
\begin{equation}\label{eq:ode}
	\frac{\mathrm{d}\rho(c)}{\mathrm{d}c}= -\frac{\gamma}{\delta}\left(\frac{1}{c} - \frac{\bar c^\delta}{c^{1+\delta}}\right) \rho(c).
\end{equation} 
The unique solution of \eqref{eq:ode} is then given by
\begin{equation}\label{eq:equi}
	h_\infty(c) = h_\infty(\bar c)\left(\frac{\bar c}{c}\right)^{2+\delta+\gamma/\delta}\text{exp} \left\{ -\frac{\gamma}{\delta^2}\left(\left(\frac{\bar c}{c} \right)^\delta -1 \right)\right\},
\end{equation} 
which is known as {\em Amoroso distribution}, corresponding to a particular class of the generalized Gamma distribution. 

We focus now on a particular case, i.e. the case in which $\delta \to 0$ and consequently $\beta=0$, i.e. the collision kernel is independent on the number of contacts. In this situation, the value function \eqref{eq:vfeps} degenerates to
\begin{equation}\label{eq:vf0}
	\Psi^\epsilon_0(s) =- \mu \frac{s^{-\epsilon} - 1}{(1-\mu)s^{-\epsilon} + 1 + \mu}=\left(\frac{\mu}{1-\mu}\right) \frac{s^{\epsilon}-1}{\frac{1+\mu}{1-\mu}s^{\epsilon} +1},
\end{equation} 
where now the following limit holds true
$$
\lim_{\epsilon \to 0} \frac{1}{\epsilon}\Psi^\epsilon_0\left( \frac{c}{\bar c}\right) =\frac{\mu/(1-\mu)}{\frac{1+\mu}{1-\mu}+1} \text{ln}\left( \frac{c}{\bar c}\right)=\frac{\mu}{2} \text{ln}\left( \frac{c}{\bar c}\right),
$$
so that the evolution of the observables, as $\epsilon \to 0$ and in a case of a collision kernel which does not depend on the number of connections, is well described by the following Fokker-Planck type equation in weak form
\begin{equation}\label{eq:fpcon0}
	\frac{\mathrm{d}}{\mathrm{d}t} \int_{\mathbb{R}_+} \varphi (c) h (c,t) \mathrm{d}c = \int_{\mathbb{R}_+}\left( -\varphi'(c)\frac{\tilde \mu}{2} \text{ln}\left( \frac{c}{\bar c}\right)c + \frac{\tilde\nu^2}{2} \varphi''(c) c^2 \right)h(c,t) \mathrm{d}c,
\end{equation} 
where $\tilde \mu = \mu/\kappa$ and $\tilde \nu^2 = \nu^2/\kappa$. Under again the zero flux hypothesis at the boundary, one gets a strong form of a Fokker-Planck type of equation 
\begin{equation}\label{eq:fpstrongcon0}
	\frac{\partial h(c,t)}{\partial t} = \frac{\tilde \mu}{2} \frac{\partial}{\partial c} \left( c \text{ln} \left(\frac{c}{\bar c}\right)h(c,t)\right) + \frac{\tilde \nu^2}{2} \frac{\partial^2}{\partial c^2}\left(c^2h(c,t)\right),
\end{equation} 
which equilibrium state now reads
\begin{equation}\label{eq:logn}
	h_\infty(c) = \frac{1}{\sqrt{2 \pi \sigma} c}\text{exp} \left\{ -\frac{(\text{ln} c - \lambda)^2}{2\sigma}\right\},
\end{equation} 
where  $\gamma = \tilde \mu/\tilde \nu^2$ and where we denoted $\sigma = 1/\gamma$ and $\lambda = \text{ln}\bar c - \sigma$. Equation \eqref{eq:logn} is a lognormal probability distribution with mean and variance respectively given by
\[
m(h_\infty) = \bar c e^{-\sigma/2}, \quad \text{Var}(h_\infty) = \bar c^2(1-e^{-\sigma}).
\]

\subsection{Contact distribution on Twitter and fitting}\label{dataextr}
We discuss now the capability of our contact model to describe real networks. To that aim, we first collected data from Twitter in order to reconstruct a typical ensemble of connections. Successively, we estimated the parameters appearing in the general equilibrium state derived previously, \eqref{eq:equi} and \eqref{eq:logn}, in such a way for our model to be as close as possible to real observable networks. The choice of using Twitter among other possible online platforms is motivated by the fact that one key characteristic of Twitter is to be more focused on staying informed and updated with respect, for instance, to Facebook which mainly aims at making friends. Thus, the first seemed more adapted to the study of opinion formation and modification with respect to the latter. 

At the time we started this research, Twitter allowed academics to retrieve information about its users, given the IDs or usernames. Thanks to this possibility, the data have been collected through the user IDs of some of the most followed politicians from around the world, and successively by accessing their network, we retrieved the IDs of a million followers from each of their profiles. We then merged all these data, ignoring possible repetitions, and extracted one million profiles from such a set. 
We then gathered the number of followers of each participant in order to have a statistical representation of the distribution of connections over the platform. 
During this operation, we eliminated profiles with zero followers, assuming them as inactive. In Figure \ref{fig:network} we represent our data set using a sample of $N_{\texttt{s}}=400$ accounts from the $N=10^6$ accounts extracted from Twitter. Sizes of the bubbles are proportional to the logarithm of agents' contacts, where the information of the edges connection is reconstructed based on the statistical distribution of the connections.
\begin{figure}[h!]
	\centering
	\includegraphics[width=1\textwidth]{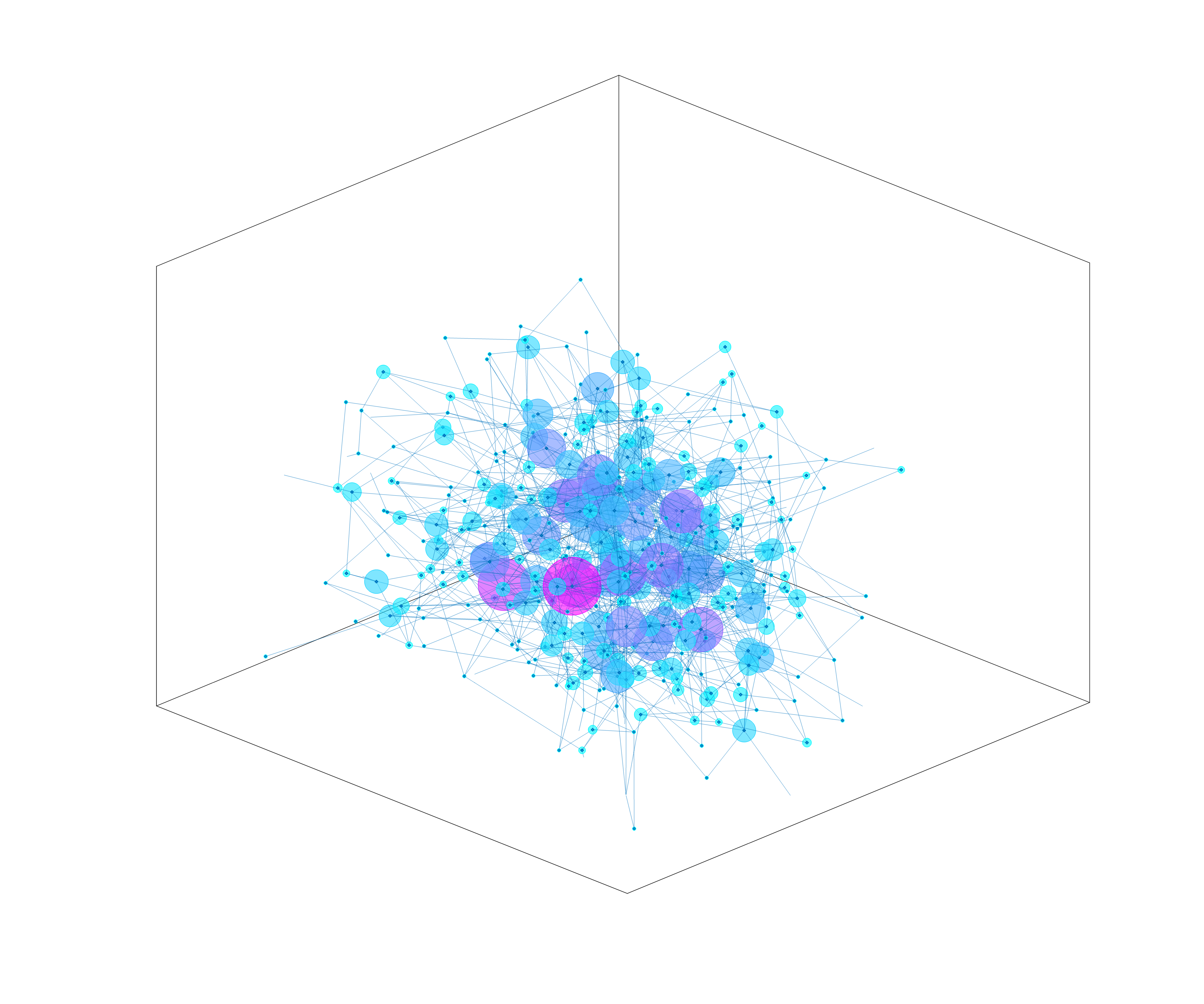} 
	\caption{Representation of the social network using a sample of $N_{\texttt{s}}=400$ accounts from the $N=10^6$ data-set extracted from Twitter. Sizes of the bubbles are proportional to the logarithm of agents' contacts,where edges are reconstructed based on the statistical distribution of the connections}
	\label{fig:network}
\end{figure}
The fitting of the contact distribution arising from the data with the steady state solution \eqref{eq:equi}-\eqref{eq:logn} has been obtained by solving a nonlinear least-squares problem, through the Matlab function \verb|lsqcurvefit|. The analysis of the best fit has been done using different choices for the contact distribution, namely lognormal, Amoroso and Inverse Gamma distributions have been tested (each one corresponding to different values of the parameter $\delta$ appearing in the value function). In Table \ref{tab:t1}, we show the results of such a study obtained with different fitting functions. We can conclude that the best fit is obtained in the case in which the steady-state distribution of the number of social media connections is distributed as a {\em lognormal density}, meaning that the parameter $\delta$ in Section \ref{know5} is equal to $0$.
\begin{figure}[h!]
	\centering
	\includegraphics[width=0.9\textwidth]{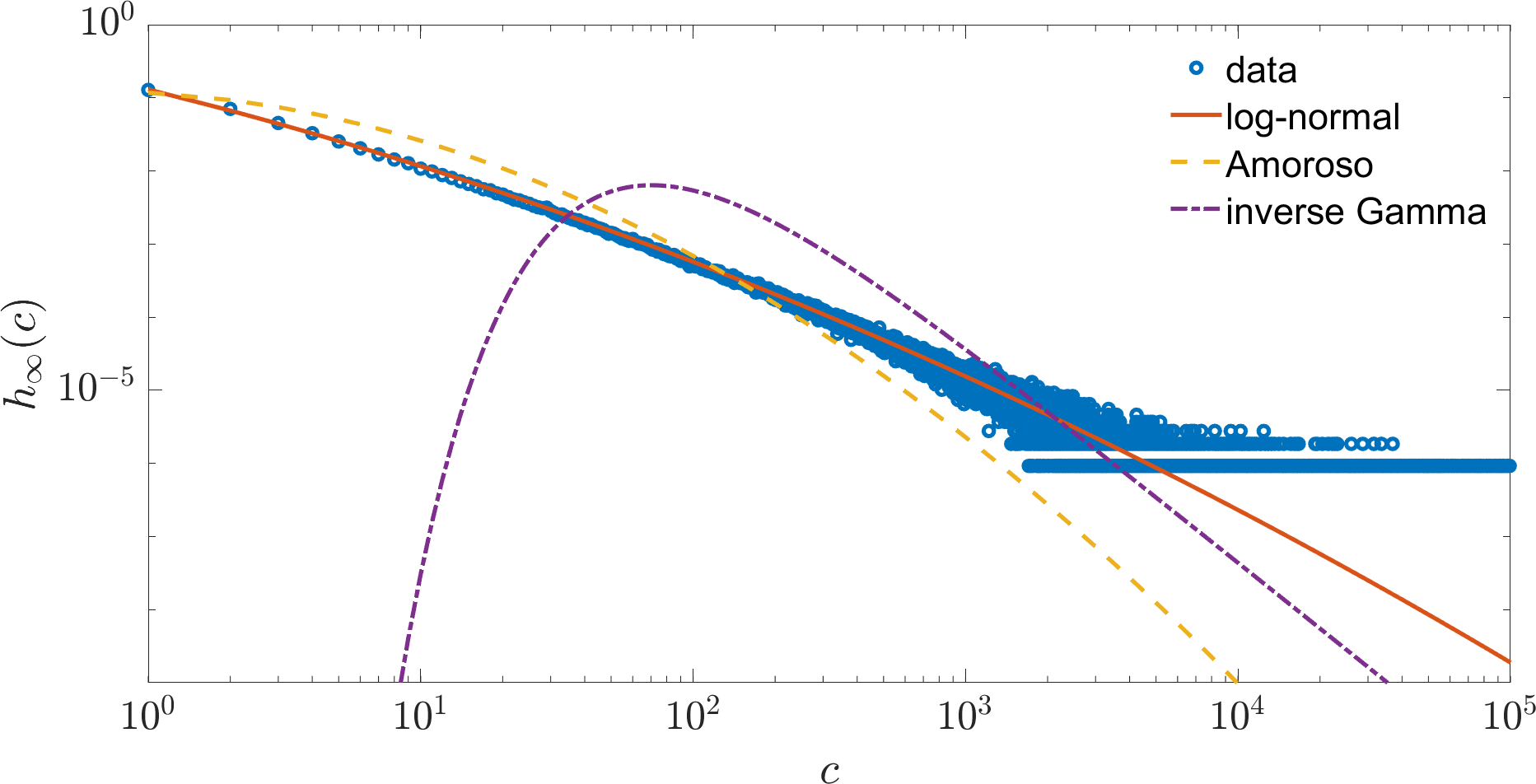} 
	\caption{Comparison between the tails of the data distribution and the different possible equilibrium distributions of the Fokker-Planck models of Section \ref{kinetic_model}}
	\label{fig:code}\end{figure}
\begin{table}[h!]
	\centering
	\begin{tabular}{|c|c|c|c|c|} 
		\hline
		Distribution type & $\delta$ & $\bar c$ & $\gamma$ & Error   \\
		\hline 
		Log-normal distribution & $0$& $40$ & $3.36\cdot 10^{-1}$ & $2.96\cdot 10^{-2}$   \\
		\hline
		Amoroso distribution & $6.56\cdot 10^{-3}$ & $94$ & $3.93\cdot 10^{-1}$ &  $4.68\cdot 10^{-1}$   \\
		\hline 
		Inverse Gamma distribution &  $1$ & $2.11\cdot 10^{6}$ & $1.00\cdot 10^{-4}$ & $1.04$ \\
		\hline 
	\end{tabular}\vspace{0.1cm}
	\caption{\label{tab:t1} Fitting of the contact distribution from Twitter-data}
\end{table}

We should remark that the log-normal distribution only has two parameters, while both the Amoroso and the Inverse Gamma have three parameters, and that the process of fitting is easier when fewer parameters have to be identified. Notice also that both the mean and the variance of the steady state depend on $\nu$, the ratio between the variance of the random percentage of variation of followers, and $\mu$, the maximal percentage allowed of possible variation of followers per interaction. Referring to \eqref{eq:logn}, the values of the parameters resulting from the fitting process are $\lambda = 7.165 \times 10^{-1}$ and $\sigma = 8.882$. In the sequel, we will then use a lognormal distribution to characterize the structure of the network.

\section{Kinetic model of opinions and contacts}\label{model}
As done for the process of evolution of the number of social media contacts, we again start from the microscopic interactions between individuals interacting on a social platform to model the evolution of the distribution of opinions in time, as done for instance in \cite{To06}. However, here we take into account the possibility that opinions of people having a large number of connections have a larger impact on the community and that consequently they more easily modify other opinions.

\subsection{The binary interaction}
We start by associating the opinion of each agent with a variable $v \in I = [-1,1]$. At the microscopic scale, we then suppose that binary interactions between individuals obey the following law
\begin{equation}
	\begin{split}
		\label{eq.trules}
		v' &     =  v + \alpha P(v,v_*,c,c_*)(v_*-v) +\xi D(v,c),\\
		v'_* & =  v_* + \alpha P(v_*,v,c_*,c)(v-v_*)+\xi_* D(v_*,c_*),
	\end{split}
\end{equation}
where $v$ and $v_*$ are the agents' opinions before the interaction, while $v'$ and $v'_*$ are their opinions after interacting. In \eqref{eq.trules}, the function $P$ can be seen as the \emph{compromise propensity} of the agent. In other words, as a consequence of the exchange of relative information, the two interacting agents change their opinions, in a symmetric or more in general non-symmetric way, approaching one the opinion of the other and vice versa. The function $D$ instead is responsible for diffusion effects and it models the unpredictable role played by the environment. It is indeed multiplied by the random variable $\xi$, with $\langle \xi  \rangle =0$, and  $\langle \xi^2 \rangle =\sigma^2$. Let us observe that in the general case depicted in \eqref{eq.trules} the post-interaction opinions depend upon the number of connections of both participants to the interaction. We will detail this dependence later. We remark for the moment that the restriction of the binary interaction to the case in which the values assumed by the couple $(v',v'_*)$ are independent of the number of contacts $(c,c_*)$ can be considered quite classical and it is discussed for instance in \cite{APTZ17}. 

Let now $f(v,c,t)$ be the density of agents which at time $t >0$ are represented by their opinion $v$ and have connection $c$. The time evolution of the distribution of opinions/connections $f(v,c,t)$, consequence of the binary interactions of type \fer{eq.trules} among individuals acting on a social platform, is obtained by resorting to kinetic collision-like models \cite{To06,PT13}. This reads in weak form as
\begin{equation} \begin{aligned}
		\label{kine-ww}
		& \displaystyle\frac{d}{dt}\int_{I \times\R_+}f(v,c,t)\varphi(v,c)\,dv\,dc  =\displaystyle \frac 12
		\Big \langle \int_{I^2\times\R_+^2} \bigl(\varphi(v',c')+ \varphi(v'_*,c'_*)\cr
		&\qquad\qquad\qquad\qquad \displaystyle-\varphi(v,c)-\varphi(v_*,c_*) \bigr) f(v_*,c_*,t)f(v,c,t)
		\,dv\,dv_*\,dc\,dc_* \Big \rangle.
	\end{aligned}
\end{equation}
In \fer{kine-ww}, the post-interaction opinions $v'$ and $v'_*$  are given by  \fer{eq.trules} while the post-interaction connections are given by \eqref{k1}. The operator $\langle \cdot \rangle$ represents the mathematical expectation with respect to the random variables $\xi$ and $\eta$. Let's observe that here we do not consider an interaction kernel depending on the number of contacts as done for instance in \eqref{eq:ble}. This choice is driven by the fact that we aim to represent a specific situation when comparing the model to the experiments, namely the case in which the network is well described by a lognormal distribution as shown in Section \ref{dataextr}. However, we stress that the extension to the case of kernels depending on $c$ is possible even if not discussed in the present work.

The opinion variable $v$ belongs to the bounded domain $[-1,1]$, so it is important to only consider interactions that do not produce values outside of such domains. A sufficient condition to preserve the bounds is given by the following proposition.
\begin{proposition}
	The binary interaction \eqref{eq.trules} preserves the bounds, i.e. $v',v'_*\in[-1,1]$ if $v,v_*\in[-1,1]$ and if 
	\begin{equation}\label{bounds}
		0 < P(v,v_*,c,c_*) \leq 1, \quad 0<\alpha \leq 1/2, \quad |\xi| \leq(1-\gamma^*)d 
	\end{equation} 
	where
	\begin{equation}\label{bounds2}
		\gamma^* = \alpha \min_{\substack{v,v_* \in [-1,1], \\c,c_*>0}} P(v,v_*,c,c_*), \quad d = \min_{\substack{v\in[-1,1], \\ c>0}} \left\{ \frac{1-|v|}{D(v,c)}, D(v,c) \neq 0\right\}.
	\end{equation} 
\end{proposition}
\begin{proof}
	Let us define $\gamma = \alpha P(v,v_*,c,c_*)$. We first consider the case in which there is no diffusion, i.e. $\xi = 0$: we have that
	$$
	|v'| = | v + \gamma(v_* - v)| \leq (1-\gamma)|v| + \gamma |v_*| \leq 1,   
	$$
	since $|v|, |v_*| \leq 1$ and, under the hypothesis \eqref{bounds}, we have that $0<\gamma \leq 1$.
	
	Let us now assume that $\xi \neq 0$: we can write, using that $|v_*|\leq 1$,
	\[
	|v'| = |v + \gamma(v_* - v) + \xi D(v,c)| \leq (1-\gamma)|v| + \gamma |v_*| + |\xi| D(v,c) \leq  (1-\gamma)|v| + \gamma + |\xi| D(v,c).
	\]
	So, in order to have $|v'|\leq 1$ it is sufficient to require
	\[
	|\xi| \leq \frac{(1-\gamma)(1-|v|)}{D(v,c)},
	\]
	with $D(v,c)\neq 0$, for all the possible values of $v$ and $c$. Thus, defining $\gamma_*$ and $d$ as in \eqref{bounds}, we get the result.
\end{proof}

\subsection{Fokker-Planck asymptotics}\label{FPdes}
In order to model the fact that the formation of the opinions are due to a large number of interactions, each one producing a small change in the point of view of individuals up to the moment in which the final opinion is formed, we regularize the dynamics of the Boltzamnn-like equation \eqref{kine-ww} by relying on a quasi-invariant scaling. This computation permits us to get some insights on the behavior of the model \eqref{kine-ww} by retrieving a Fokker-Planck equation for the asymptotic combined evolution of contacts and opinions. The quasi-invariant scaling is as follows
\begin{equation}\label{eq:scaling}
	\alpha\to\e\alpha,\qquad \sigma^2  \to \e \sigma^2,
\end{equation} 
for $\e\ll1$, similarly to the relation \eqref{k1eps} for the sole contacts dynamics.

We assume that the scaled random variables $\eta_\e$, $\xi_{\e}$, and $\xi_{\e*}$ are independent with zero mean and bounded moments at least of order $n=3$. We also assume that $\xi_{\e},\xi_{\e*}$ are identically distributed, and that the following  holds
\begin{equation}
	\langle\xi_{\e}\rangle =\langle\xi_{\e*}\rangle= 0,\quad \langle\xi_{\e}^2\rangle =\langle\xi_{\e*}^2\rangle = \e\sigma^2,\quad  \langle\xi_{\e}^3\rangle=\langle\xi_{\e*}^3\rangle = \e^{3/2}\varrho,
\end{equation} 
and  for $\eta_\e$ we have recalled the following
\begin{equation}
	\langle\eta_{\e}\rangle= 0,\quad \langle\eta_{\e}^2\rangle = \e\nu^2,\quad  \langle\eta_{\e}^3\rangle= \e^{3/2}\varkappa.
\end{equation} 
with $\varrho$ and $\varkappa$ two assigned constants. To ease the notation we now rewrite the value function \eqref{eq:vf} multiplied by the number of contacts as follows
\begin{equation}\label{eq:valuenew}
	c\Psi_\delta^\epsilon(c/\bar c) = \mu {L}_\epsilon(c),  
\end{equation} 
and the interaction function in \eqref{eq.trules} as 
\[
E(v,v_*,c,c_*) = P(v,v_*,c,c_*)(v_*-v).
\]
Using the previous properties of the random quantities $\xi, \xi_*, \eta$, the equations \eqref{k1} for the contacts and  \eqref{eq.trules} for the opinions we have the following 
\begin{equations}
	\label{var1}
	& \langle c' -c\rangle = \langle-\mu {L}_\epsilon(c) + \eta_\epsilon c \rangle= -\mu {L}_\epsilon(c),
	\\
	&\langle v' -v \rangle = \langle \alpha P(v,v_*,c,c_*)(v_*-v) +\xi D(v,c) \rangle = \e\alpha E(v,v_*,c,c_*),
\end{equations}
and

\begin{equations}
	\label{var2}
	&  \langle (c' - c)^2 \rangle =\mu^2 {L}_\epsilon(c)^2 + \epsilon \nu^2 c^2,
	\\[6pt]
	& \langle (v' -v)^2 \rangle =  \e^2\alpha^2 E(v,v_*,c,c_*)^2 + \e\sigma^2 D^2(v,c),
	\\[6pt]
	& \langle (c' -c)(v'-v) \rangle = -\e\alpha\mu {L}_\epsilon(c) E(v,v_*,c,c_*)
\end{equations}
while the third order terms are
\begin{equations}
	\label{var3}
	&  \langle (c' - c)^3 \rangle = -\mu^3 {L}_\epsilon(c)^3 + \e^{3/2}\varkappa c^3 - 3\epsilon \mu \nu^2  L_\epsilon(c) c^2 ,
	\\[6pt]
	&  \langle (v' -v)^3 \rangle =  \e^3\alpha^3E(v,v_*,c,c_*)^3 +  \e^{3/2}\varrho D^3(v,c) +3\e^2 \alpha\sigma^2 E(v,v_*,c,c_*) D(v,c)^2,
	\\[6pt]
	& \langle (c' -c)^2(v'-v) \rangle  =  \e\alpha E(v,v_*,c,c_*)(\mu^2 L_\epsilon(c)^2 + \epsilon \nu^2 c^2),
	\\[6pt]
	&  \langle (c' -c)(v'-v)^2 \rangle =-\mu L_\epsilon(c)( \e^2\alpha^2E(v,v_*,c,c_*)^2 +\e \sigma^2 D^2(v,c) ).
\end{equations}
By expanding the smooth function $\varphi(x^*,v^*)$ in Taylor series up to
order two we have
\begin{equation}\label{eq:Tayex}
	\begin{aligned}
		&\langle \varphi(v',c')-\varphi(v,c) \rangle =
		\cr
		&\quad\displaystyle \e\left( \alpha E(v,v_*,c,c_*)\frac{\partial \varphi}{\partial v} - \mu \frac{L_\epsilon(c)}{\e}\frac{\partial \varphi}{\partial c} + \frac 12 \sigma^2 D(v,c)^2 \frac{\partial^2 \varphi}{\partial v^2} + \frac 12 \nu^2 c^2 \frac{\partial^2 \varphi}{\partial c^2} \right)
		\cr
		&+ \frac {\e^2}2 \left( \alpha^2E(v,v_*,c,c_*)^2 \frac{\partial^2 \varphi}{\partial v^2} +  \mu^2  \frac {L_\epsilon(c)^2}{\e^2} \frac{\partial^2 \varphi}{\partial c^2} - \mu \alpha \frac{L_\epsilon(c)}{\e} E(v,v_*,c,c_*)\frac{\partial^2 \varphi}{\partial v \partial c} \right)\cr
		&\qquad\qquad +R_\e(v,v_*,c,c_*),
	\end{aligned}
\end{equation} 
where the remainder of the Taylor expansion $R_\e(v,v_*,c,c_*)$ is expressed as follows
\begin{equation}\label{eq:Tayrem}
	\begin{aligned}
		&R_\e(v,v_*,c,c_*) = \displaystyle  \frac{\e^2}{6} \frac{\partial^3 \varphi}{\partial c^3}(\hat v, \hat c) \left( -\mu^3 \frac{{L}_\epsilon(c)^3}{\e^2} + \e^{-1/2}\varkappa c^3 - 3 \mu \nu^2 \frac{ L_\epsilon(c)}{\e} c^2 \right) \\
		& \, +   \frac{\e^2}{6} \frac{\partial^3 \varphi}{\partial v^3}  (\hat v, \hat c) \left( \e\alpha^3E(v,v_*,c,c_*)^3 +  \e^{-1/2}\varrho D^3(v,c) +3 \alpha\sigma^2 E(v,v_*,c,c_*) D(v,c)^2\right)\\
		& \quad + \frac {\e^2}{2}\frac{\partial^3 \varphi}{\partial v\partial c^2}(\hat v, \hat c)\left( \alpha E(v,v_*,c,c_*)\left(\mu^2 \frac{L_\epsilon(c)^2}{\e} + \nu^2 c^2\right)\right) \\[6pt]
		& \quad\quad  - \frac {\e^2}2 \frac{\partial^3 \varphi}{\partial v^2\partial c}(\hat v, \hat c) \left(\mu \frac{L_\epsilon(c)}{\e}\left( \e\alpha^2E(v,v_*,c,c_*)^2 +\sigma^2 D^2(v,c) \right)\right),
	\end{aligned}
\end{equation}
for $\hat v = \theta_v v'+(1-\theta_v)v$ with $\theta_v \in[0,1]$ and $\hat c=\theta_c c'+(1-\theta_c)c$ with $\theta_c \in[0,1]$.

%

Hence,  scaling the time variable $\tau = \epsilon t$ and using the expansion \eqref{eq:Tayex} in \eqref{kine-ww},  the solution $f_\epsilon$ satisfies the weak relation
\begin{equation}\label{eq:weakfp}
	\begin{aligned}
		&\displaystyle\frac{d}{d\tau}\int_{I \times\R_+}f_\epsilon(v,c,\tau)\varphi(v,c)\,dv\,dc  =\displaystyle
		\int_{I\times\R_+} \left( \mathcal{E}[f_\epsilon](v,c,\tau)\frac{\partial \varphi}{\partial v} - \frac{\mu}{\epsilon} \Psi_\delta^\epsilon(c/\bar c) c \frac{\partial \varphi}{\partial c} \right. \\
		&\qquad\quad\displaystyle\left.+ \frac 12 \sigma^2 D^2(v,c) \frac{\partial^2 \varphi}{\partial v^2}  + \frac 12 \nu^2 c^2 \frac{\partial^2 \varphi}{\partial c^2}  \right)f_\epsilon(v,c,\tau)\,dv\,dc+ \mathcal {R}_\e(\varphi),
	\end{aligned}
\end{equation} 
where we introduced the following notation for the non-local operator
$$
\mathcal{E}[f_\epsilon](v,c,\tau) =\alpha \int_{I\times\R_+} E(v,v_*,c,c_*) f_\epsilon(v_*,c_*,\tau)\,dv_*\,dc_*,
$$ and where the scaled reminder is 
\begin{equation}\label{eq:remind}
	\begin{aligned}
		&\mathcal{R}_\e(\varphi) =  \frac {\e}2\int_{I^2\times\R_+^2} \left( \alpha^2E(v,v_*,c,c_*)^2 \frac{\partial^2 \varphi}{\partial v^2} +  \mu^2  \frac {L_\epsilon(c)^2}{\e^2} \frac{\partial^2 \varphi}{\partial c^2}\right.\cr
		&\left.\quad\quad - \mu \alpha \frac{L_\epsilon(c)}{\e} E(v,v_*,c,c_*)\frac{\partial^2 \varphi}{\partial v \partial c} 
		\right)f_\epsilon(v,c,\tau)f_\epsilon(v_*,c_*,\tau) \,dv\,dv_*\,dc\,dc_*\cr
		&\qquad\qquad + \frac {1}{\e}\int_{I^2\times\R_+^2} 
		R_\e(v,v_*,c,c_*)f_\epsilon(v,c,\tau)f_\epsilon(v_*,c_*,\tau) \,dv\,dv_*\,dc\,dc_*.
		%
	\end{aligned}
\end{equation} 
For $\e\to 0$ we recall that from \eqref{eq:valuenew} and \eqref{eq:Tayrem} we have the following
\[
\begin{aligned}
	&L_\e(c)\to 0,\qquad {L_\e(c)}/{\e} \to\Phi_\delta(c),\qquad
	{R_\e(v,v_*,c,c_*)}/{\e} \to 0,
\end{aligned}
\]
where
\begin{equation}\label{eq:phi}
	\Phi_\delta(c):=
	\begin{cases}
		\displaystyle\frac{\mu}{2\delta} \left(1- \left( \frac{ c}{\bar c}\right)^{-\delta}  \right)c,\quad 0<\delta\leq 1\\\\
		\displaystyle	\frac{\mu}{2} \text{ln}\left( \frac{c}{\bar c}\right)c,\,\qquad \delta\to 0.
	\end{cases}
\end{equation} 
Hence, in the limit $\e\to 0$ the reminder \eqref{eq:remind} vanishes and the equation \eqref{eq:weakfp} collapses to the weak form of the following equation 
\begin{equation}\label{fpeq}
	\displaystyle\frac{\partial f}{\partial \tau}=-\frac{\partial   \left(\mathcal{E}[f](v,c,\tau) f \right)}{\partial v} +  \frac{\partial ( \Phi_\delta(c)f)}{\partial c}+ \frac 12 \sigma^2 \frac{\partial^2 (D^2(v,c)f)}{\partial v^2}  + \frac 12 \nu ^2 \frac{\partial^2( c^2f)}{\partial c^2}.
\end{equation} 
Equation \eqref{fpeq} is the model we will use in the sequel to describe the time evolution of opinion formation over a social network. In particular, in the final part of Section \ref{nume} focusing on a specific platform, namely Twitter, we will fit the parameter in our model with experimental data with the scope of describing a realistic phenomenon.

\subsection{On the steady state solution for the opinion distribution}
In the general case, the steady state of equation \eqref{fpeq} is not known. However, it is possible to compute an explicit formula for the asymptotic solution under some particular assumptions. Let assume that $D(v,c) = 1 - v^2$ and $P(v,v_*,c,c_*)=1$ so that
\begin{equation*}\begin{array}{rcl}
		\mathcal{E}[f](v,c,\tau) &=& \displaystyle\alpha\left(\int_{I\times\R_+} v_* f(v_*,c_*,\tau)\,dv_*\,dc_* - v \right) 
		=\alpha \left( m_v(\tau) - v \right).
	\end{array}
\end{equation*}
In this situation, one can look to solutions of type $f(v,c,t) = g(v,t)h(c,t)$ leading to asymptotic of the form $f_\infty(v,c)=g_\infty(v) h_\infty(c)$, where $h_\infty$ is the asymptotic state of the social contact distribution derived in Section \ref{know5}. Under the above hypothesis,  the Fokker-Planck equation \eqref{fpeq} can be rewritten as
\begin{equation}\begin{array}{rcl}
		\displaystyle\frac{\partial g}{\partial \tau}h + \displaystyle\frac{\partial h}{\partial \tau}g &=&-\displaystyle\alpha\frac{\partial ( (m_v(\tau) - v) g)}{\partial v}h + \frac{\mu}{2}  \frac{\partial \left(c \text{ln} \left(\frac{c}{\bar c}\right)h\right)}{\partial c}g \\[6pt]
		& & \displaystyle+  \frac 12 \sigma^2 \frac{\partial^2 ((1-v^2)^2g)}{\partial v^2} h + \frac 12 \nu^2  \frac{\partial^2( c^2h)}{\partial c^2}g,
	\end{array}
\end{equation} 
leading to
\begin{multline}\label{fpeq2}
	\left(\frac{\partial g}{\partial \tau} + \alpha\frac{\partial ( (m_v(\tau)  - v) g)}{\partial v}- \frac 12 \sigma^2 \frac{\partial^2 ((1-v^2)^2g)}{\partial v^2} \right)h \\ +\displaystyle\left(\frac{\partial h}{\partial \tau} - \frac{\mu}{2}  \frac{\partial \left(c \text{ln} \left(\frac{c}{\bar c}\right)h\right)}{\partial c} -\frac 12 \nu^2  \frac{\partial^2( c^2h)}{\partial c^2} \right)g= 0.
\end{multline}
Non-trivial solution for equation \eqref{fpeq2} are retrieved for
\begin{equation}\label{eq:perg}
	\frac{\partial g}{\partial \tau} = - \alpha\frac{\partial ( (m_v(\tau) - v) g)}{\partial v} + \frac 12 \sigma^2 \frac{\partial^2 ((1-v^2)^2g)}{\partial v^2} 
\end{equation} 
and
\begin{equation}\label{eq:perh}
	\frac{\partial h}{\partial \tau} = \frac{\mu}{2}  \frac{\partial \left(c \text{ln} \left(\frac{c}{\bar c}\right)h\right)}{\partial c} + \frac 12 \nu^2  \frac{\partial^2( c^2h)}{\partial c^2}.
\end{equation} 
Thus stationary solutions for \eqref{fpeq} are of the form $f_\infty(v,c) = g_\infty(v)h_\infty(c)$ as claimed before, where for \eqref{eq:perh} we obtain the log-normal distribution \eqref{eq:logn}, while in order to compute the stationary solution to \eqref{eq:perg}, one has to solve
\begin{equation}\label{eq:ginf}
	\frac{\mathrm{d}((1-v^2)^2g_\infty)}{\mathrm{d}v} = \frac{2\alpha}{\sigma^2}( (\bar m_v - v) g_\infty),\qquad \bar m_v =  \int_{I\times\R_+} v_* g_\infty(v_*)\,dv_*.
\end{equation} 
The solution to \eqref{eq:ginf} is given by
$$
g_\infty (v) = K_\infty(1+v)^{-2+\alpha\bar m_v/2\sigma^2}(1-v)^{-2-\alpha\bar m_v/2\sigma^2}\text{exp}\left\{ -\frac{\alpha(1-\bar m_v v)}{\sigma^2(1-v^2)}\right\},
$$
where $K_\infty$ is a normalization constant, such that the total mass of $g_\infty$ is equal to $1$. Figure \ref{fig:ginfty} shows the comparison between the analytical profile of $g_\infty(v)$ obtained in the case $\alpha = 0.1$, $\sigma^2 = 0.1$ and $\alpha = 0.25$, $\sigma^2 = 0.05$, $\bar c = 1$, $\mu = 0.1$, $\nu^2 = 0.0125$ and the numerical simulations of equation \eqref{kine-ww} through a Monte Carlo method which details are outlined in the appendix \ref{sec:appendix}. In the simulation we choose the scaling parameter $\epsilon = 0.01$ in order to retrieve the Fokker-Planck asymptotic from the Boltzmann-type equation \eqref{kine-ww}.
\begin{figure}[h!]
	\includegraphics[width=0.5\textwidth]{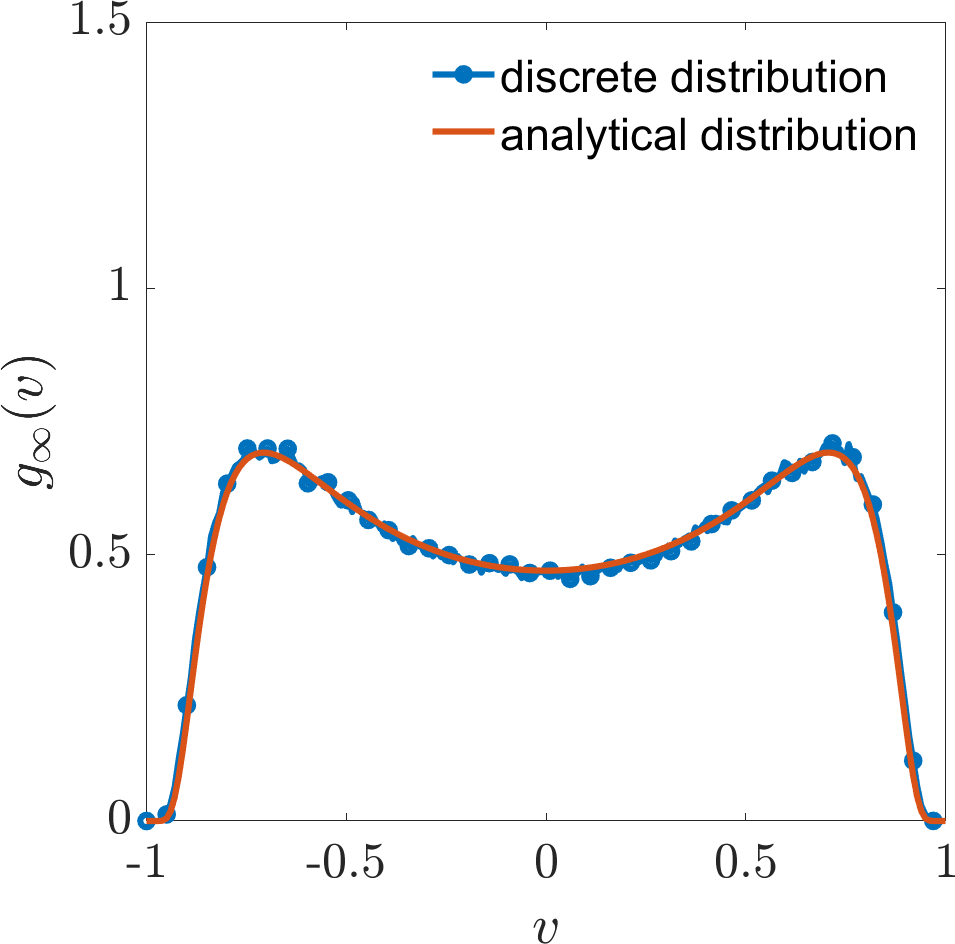} \quad\includegraphics[width=0.5\textwidth]{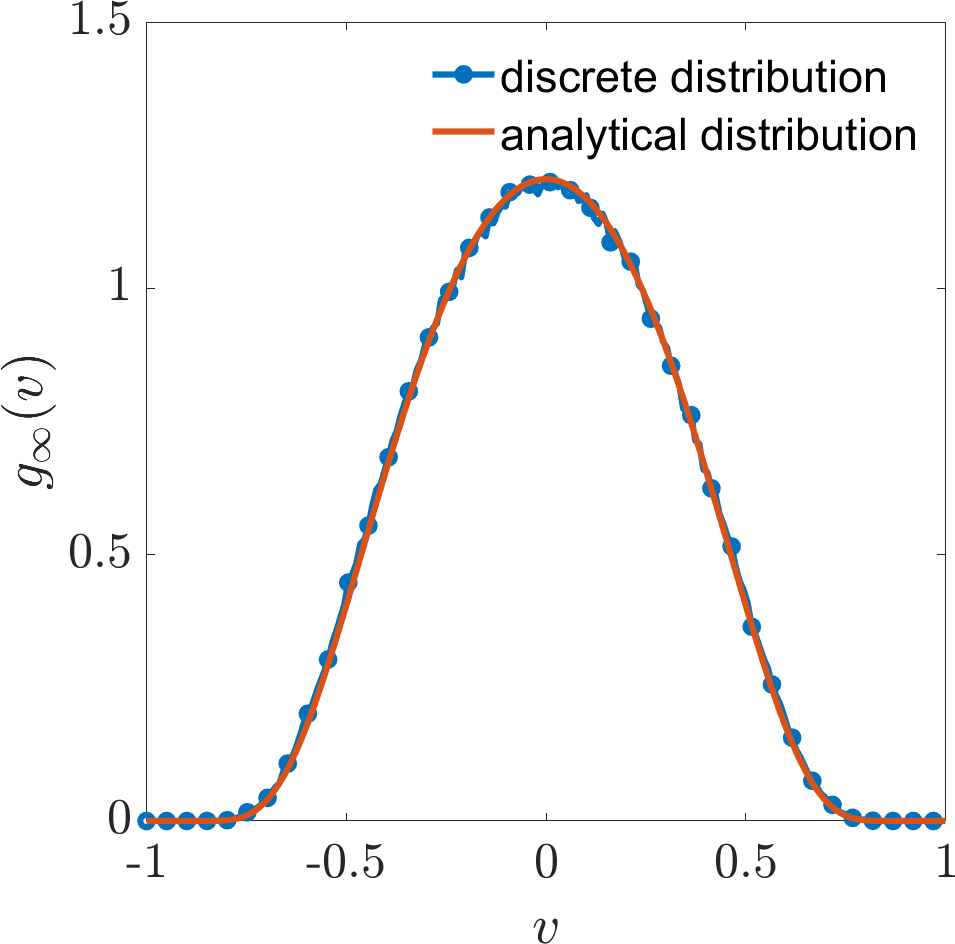} 
	\caption{Profiles of the steady state solution $g_\infty(v)$ and its numerical approximation in the case of $\sigma^2/\alpha = 1$ (left) and $\sigma^2/\alpha = 0.2$ (right), both with scaling parameter $\epsilon = 0.01$}\label{fig:ginfty}
\end{figure}

\section{Numerical experiments}\label{nume}
In order to get insights about the qualitative behavior of our new model and to validate it, we perform in the sequel different numerical simulations using a Monte Carlo-like approach for approximating the Boltzmann equation \eqref{kine-ww} in the Fokker-Planck regime \eqref{fpeq}. The details of the numerical scheme employed are given in the Appendix \ref{sec:appendix}. We will always assume, if not otherwise stated, that in the rest of the Section, the interaction kernel $P$ can be expressed as 
\begin{equation}\label{eq:kernel}
	P(v,v_*,c,c_*)=H(v,v_*,c,c_*)K(c,c_*)
\end{equation} 
for various choices of functions $H$ and $K$. Moreover, based on the experimental finding of Section \ref{dataextr}, we will restrict ourselves to the case $\delta \to 0$ in \eqref{eq:phi} which corresponds to the situation in which a log-normal distribution of connections describes the network at equilibrium.

\subsection{Qualitative behavior}
We start by discussing some qualitative behaviors of the model presented through the use of numerical simulations. We first consider a bounded confidence model, in which the propensity to consensus is influenced by the number of connections. Furthermore, in a second test, we consider a confidence bound depending also on the number of contacts, where we compare homogeneous with heterogeneous cases.
In a final test, we study a case in which a part of the population acting on a social network is composed of almost inflexible individuals. 

\subsubsection{Test 1: Bounded confidence model} 
In this first test we build up a bounded confidence model \cite{Deffuant,hegselmann2002opinion} by performing the following choices 
\begin{equation}\label{eq:kernels}
	H(v,v_*,c,c_*)=\chi_{\{|v-v_*|< \Delta\}}(v_*), \qquad K(c,c_*)=\frac{c_*^2}{c^2+c_*^2},
\end{equation} 
where $\chi(\cdot)$ is the indicator function and $\Delta$ is a positive constant. The 
diffusion part is weighted by $D(v,c,c_*) = 1-v^2$. In the chosen setting, the interactions take place only if two individuals have sufficiently close opinions. Moreover, we give a higher relevance, in driving the opinions process, to the agents having more connections, i.e. the influencers. Instead, individuals with few followers are less likely to be able to change the point of view of the other participants while prone to change their own opinion. 
In order to stress the importance of the presence of contacts and their relevance in modifying the evolution of the joint density $f(v,c,t)$, we start from two different initial data: in the first simulation the initial data is given by 
$$
f_0(v,c) = \frac{1}{2}h_\infty(c),
$$
meaning that the opinion is uniformly distributed in the interval $I=[-1,1]$, and the contacts are at the equilibrium \eqref{eq:logn} with parameters $\lambda = 5$ and $\sigma = 1.56\cdot 10^{-2}$. In Figure \ref{fig:test2bis}, the initial data is shown in the first image, followed by the images of the resulting density $f(v,c,t)$ for different times, namely $t = 4, 8, 12, 16, 20$. We sample $N_s=10^5$ agents to simulate model \eqref{kine-ww}, with scaling parameter $\epsilon = 0.01$. In Figure \ref{fig:test2bis} we clearly see the segmentation of the opinion which is due to the presence of a bounded confidence interaction with $\Delta = 0.55$, but over time the two clusters merge again, the consensus is reached and we notice the formation of a single cluster in $v=0$.

In the second simulation, we suppose that contacts are still at the equilibrium \eqref{eq:logn} with the same parameters $\lambda$ and $\sigma$ as in the previous simulation, but the opinions are now distributed so that agents with a low number of connections have an opinion closer to $-1$, while agents with a higher number of contacts have an opinion closer to $+1$. In Figure \ref{fig:test2bis2}, the initial data is shown in the first image, followed by the images of the resulting density $f(v,c,t)$ for $t = 4, 8, 12, 16, 20$. We see the emergence of two clusters but the symmetry of the previous case is lost: agents with a lower number of contacts are influenced by agents with a higher bound of contacts and behave like followers, whereas agents with a higher bound of contacts are less influenced by lower-contact agents. We stress that the dynamics observed in these two cases are different from the standard dynamics obtained with a bounded confidence opinion model where connections do not play a role. In fact, in this case, the steady state solution is represented by a bimodal distribution similar to the one obtained in Figure \ref{fig:test2bis2} for $t=8$.
\begin{figure}[h!]
	\includegraphics[width=0.3\textwidth]{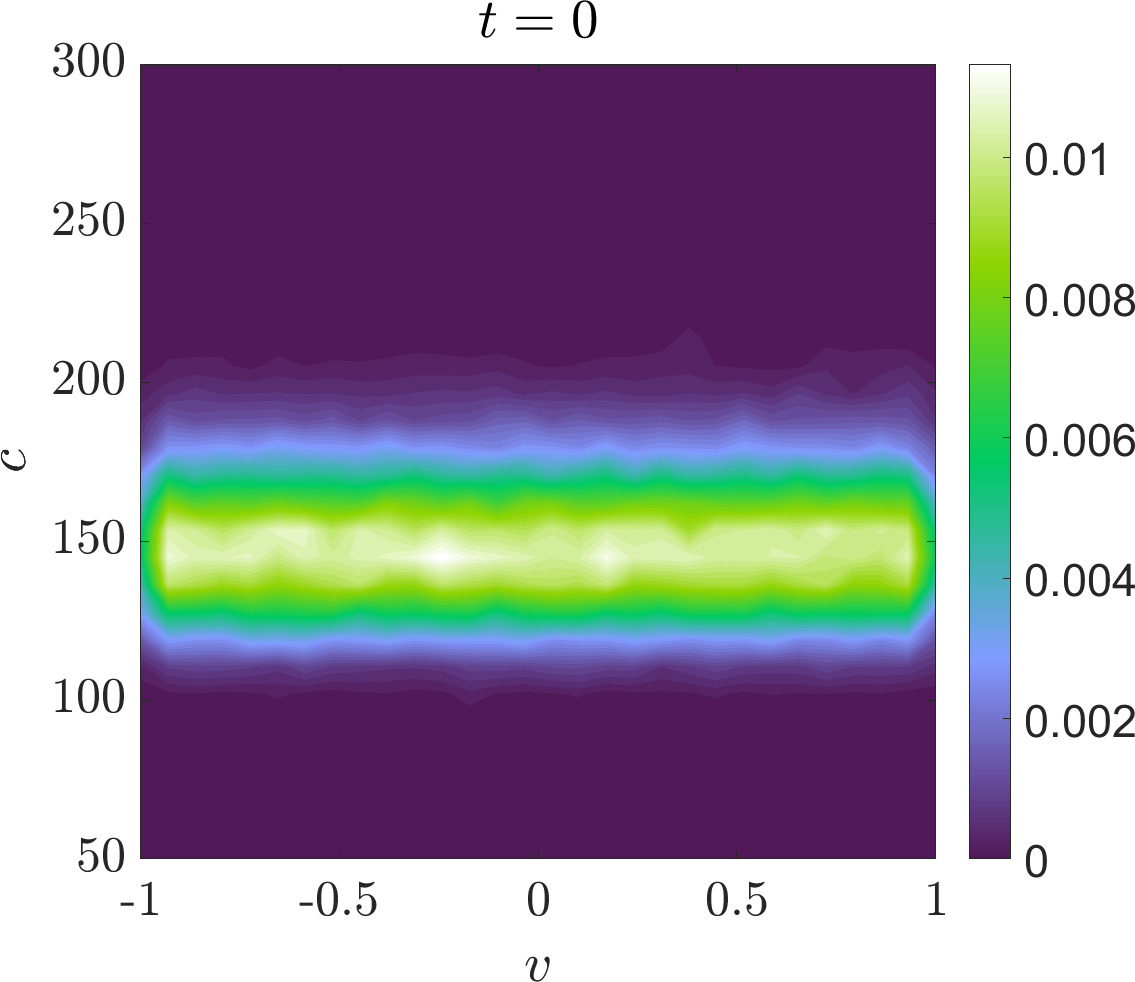} \quad\includegraphics[width=0.3\textwidth]{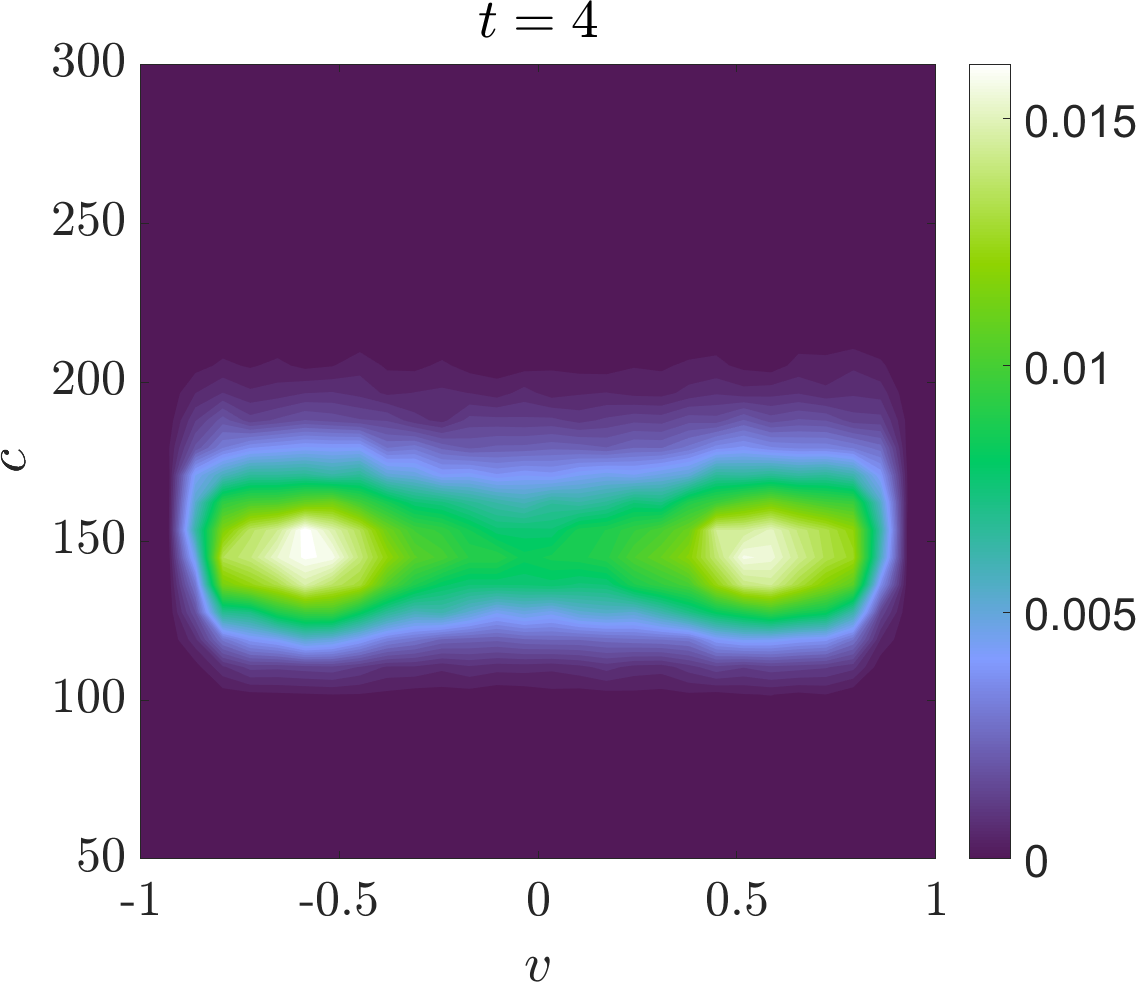} \quad\includegraphics[width=0.3\textwidth]{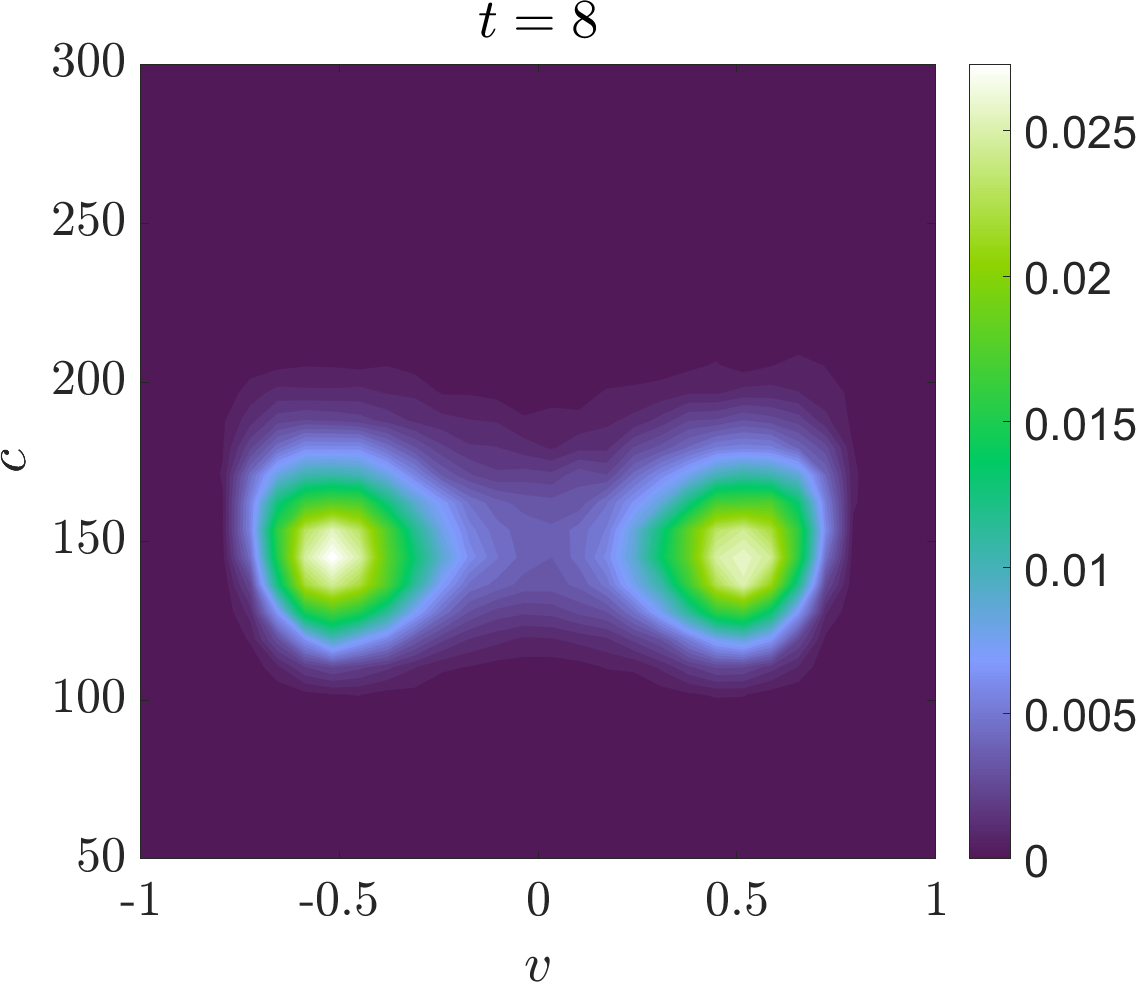} 
	\newline
	\newline
	\includegraphics[width=0.3\textwidth]{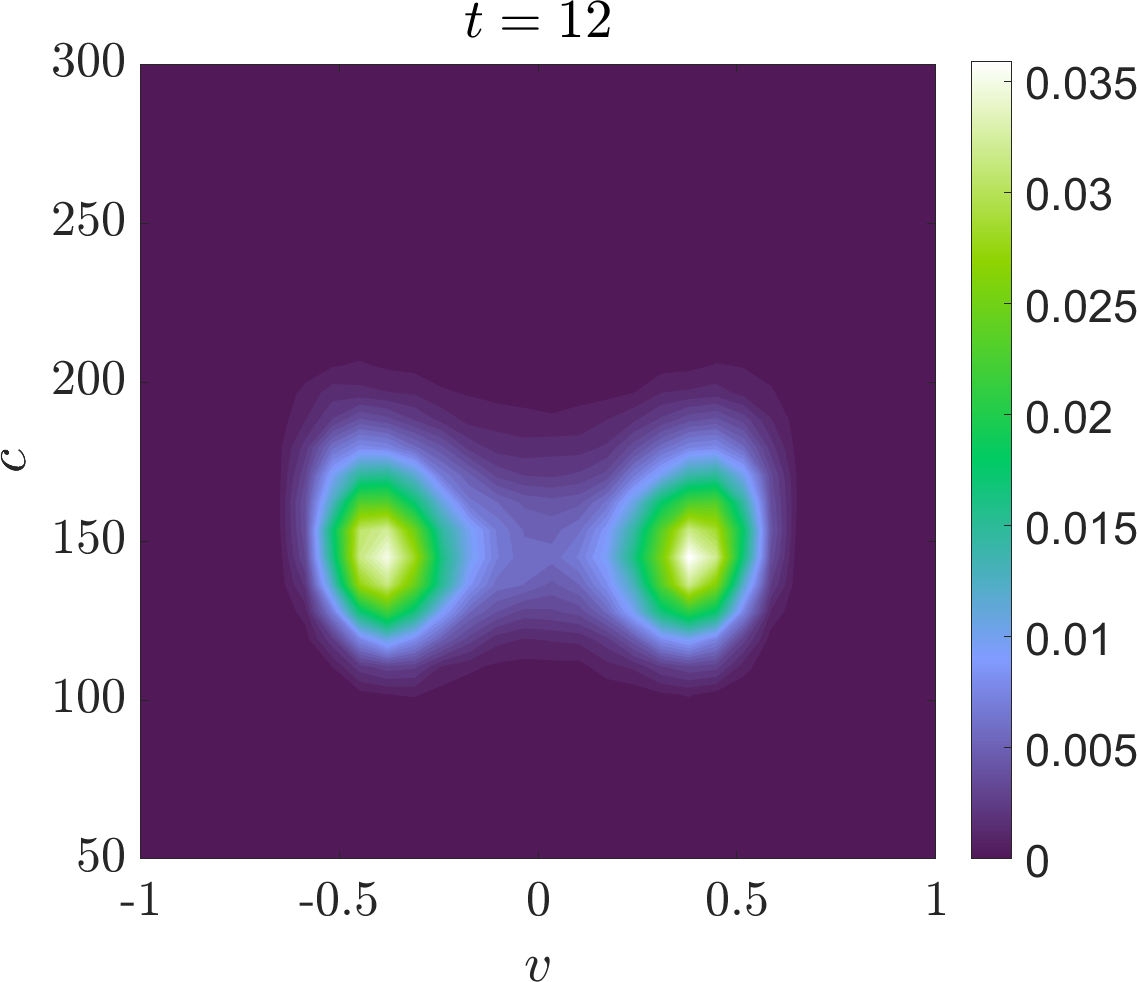} \quad\includegraphics[width=0.3\textwidth]{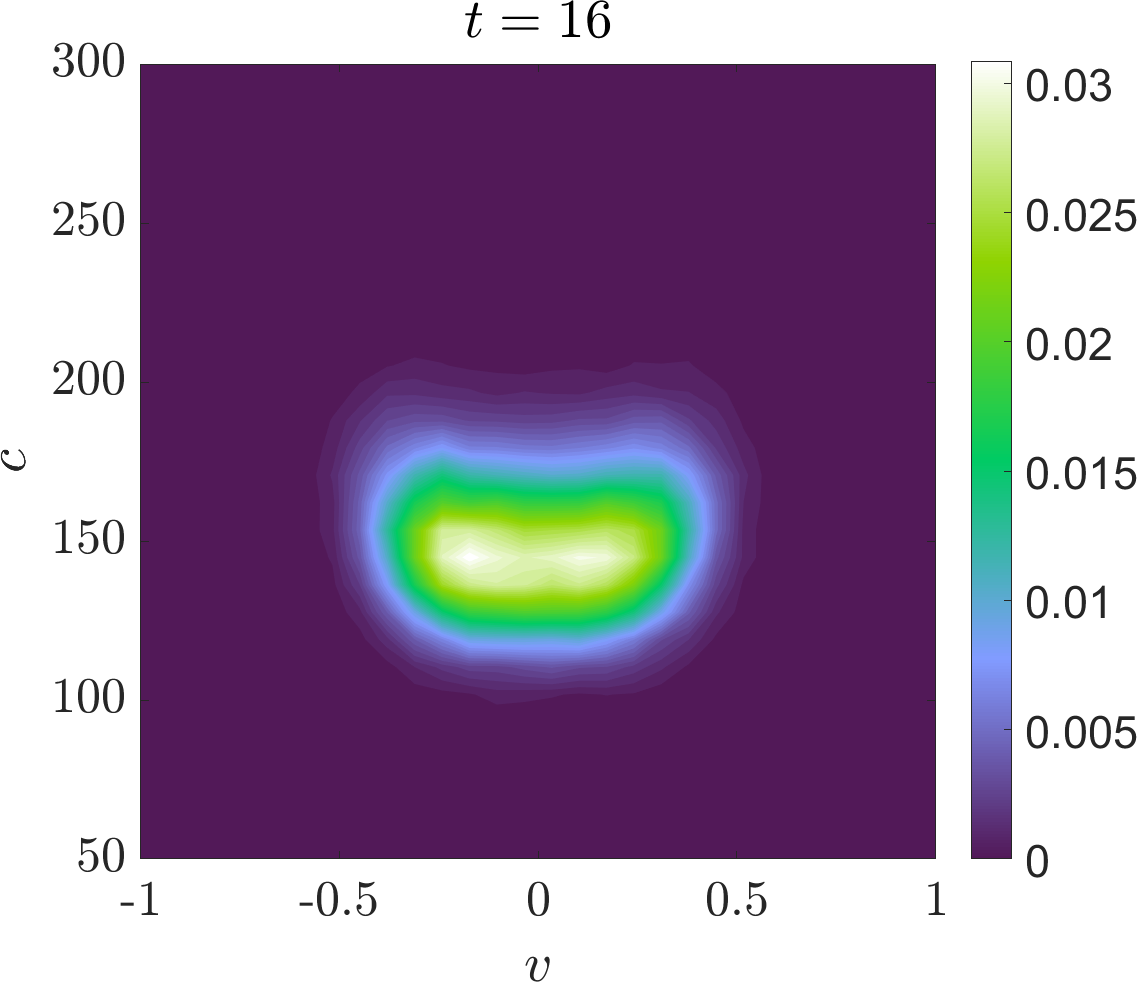} \quad\includegraphics[width=0.3\textwidth]{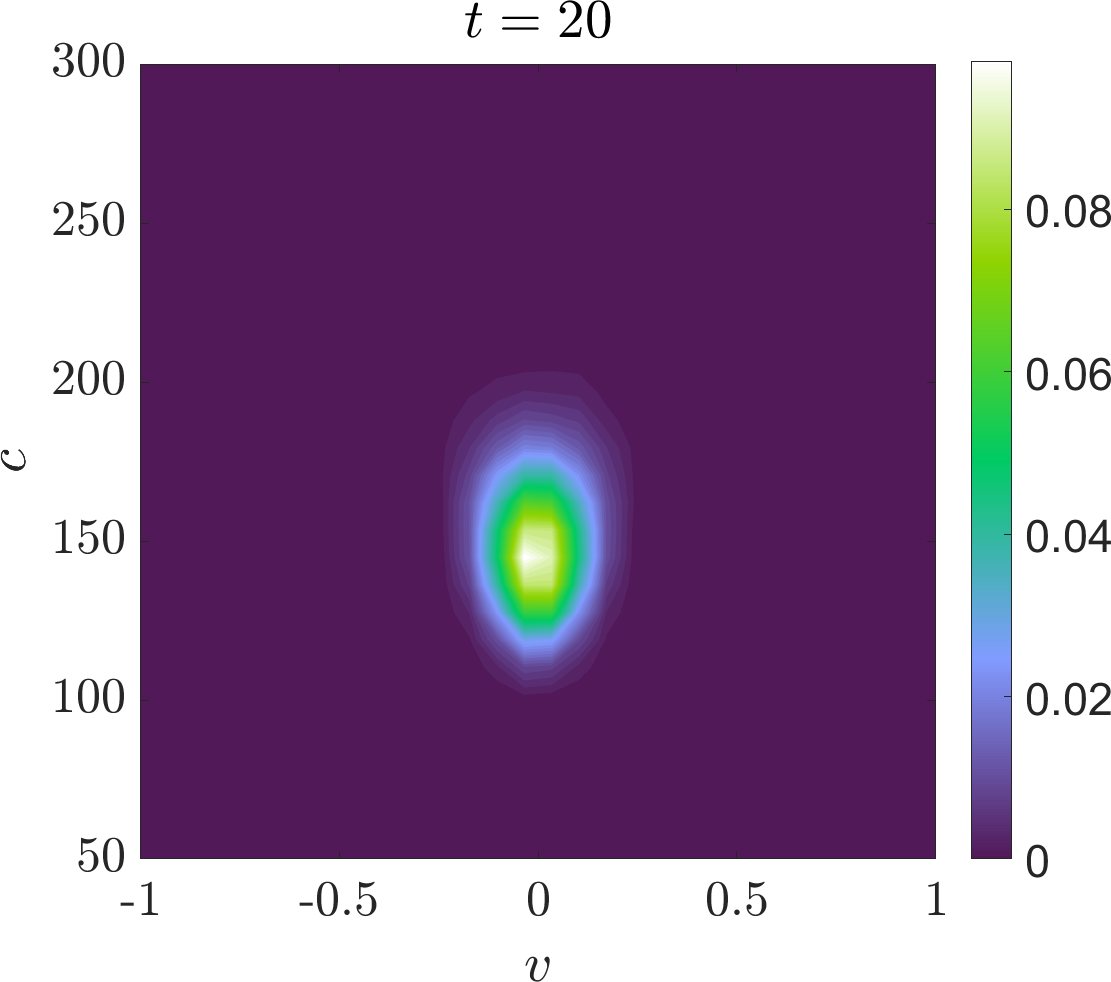} 
	\caption{Test $1$, $\sigma^2/\alpha = 0.005$. The pictures show the time evolution of the distribution function $f(v,c,t)$ for $t=0,4,8,12,16,20$ for a homogeneous distribution of the number of connections with respect to opinions. After the emergence of two clusters the agents reach consensus at the final time
	}\label{fig:test2bis}
\end{figure}

\begin{figure}[h!]
	\includegraphics[width=0.3\textwidth]{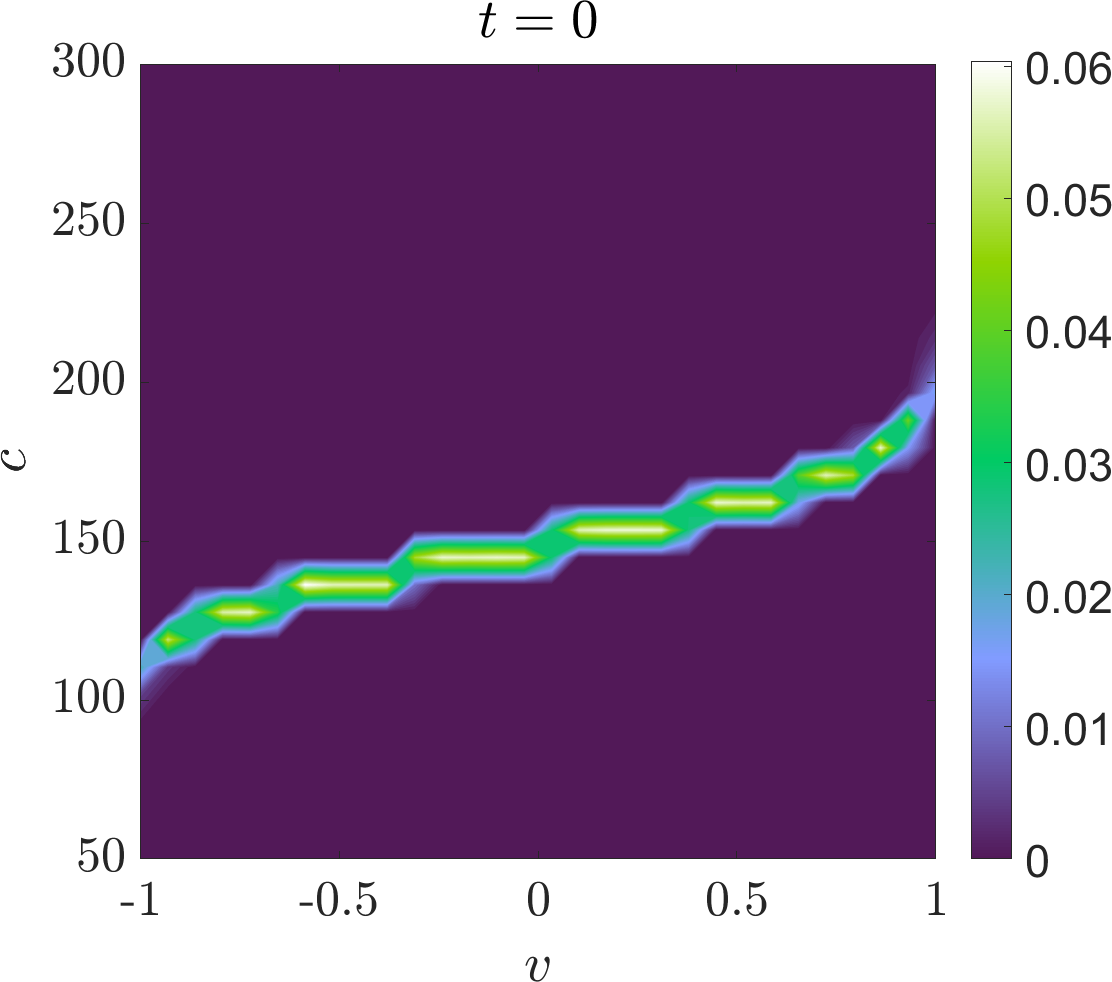} \quad\includegraphics[width=0.3\textwidth]{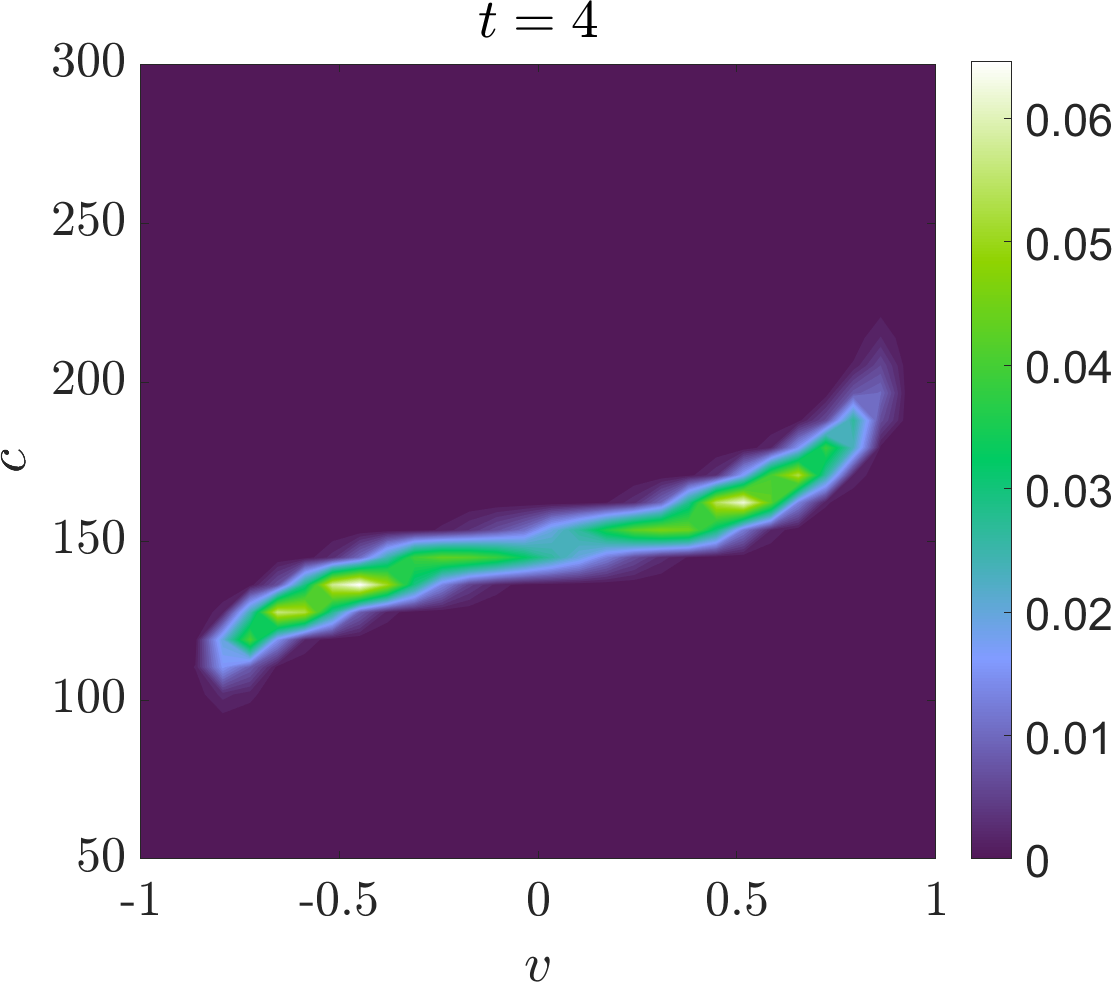} \quad\includegraphics[width=0.3\textwidth]{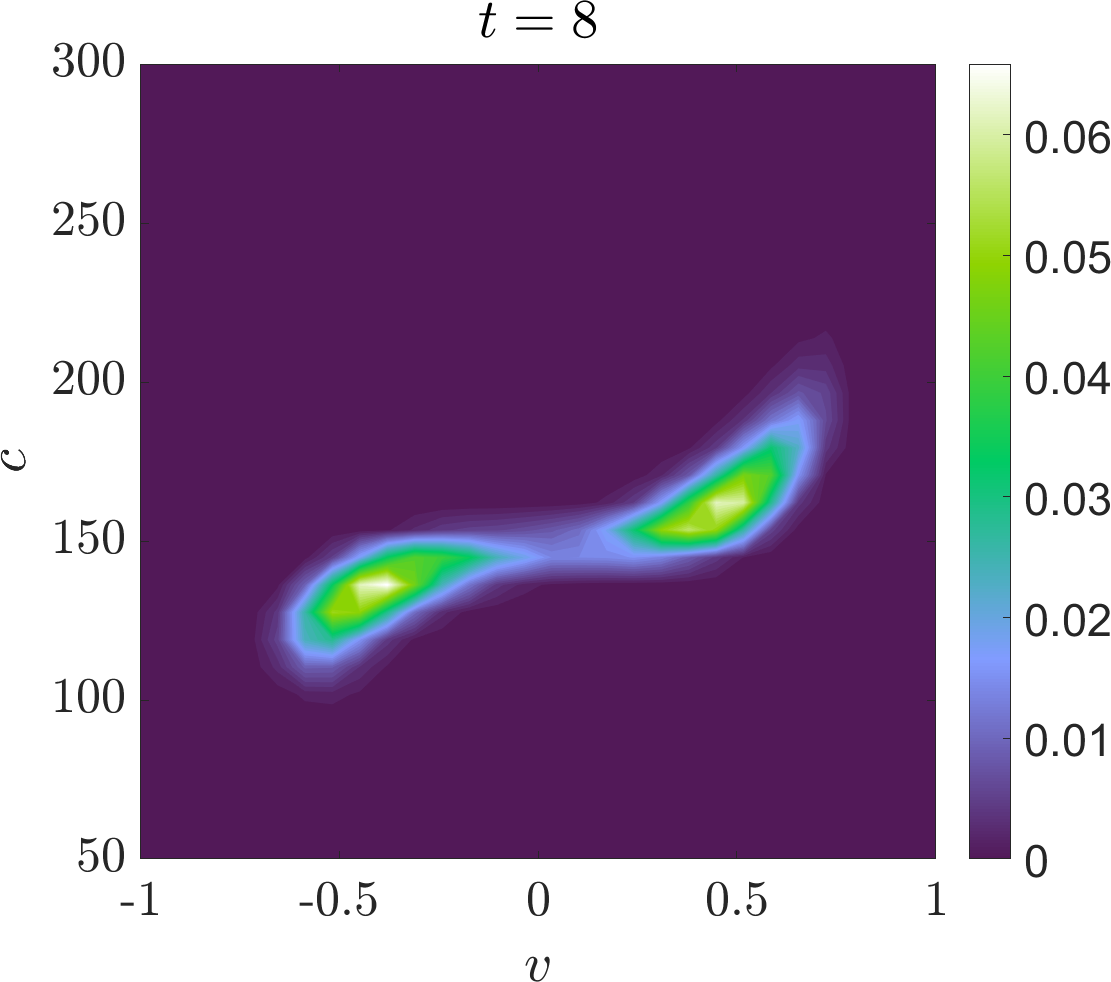} 
	\newline
	\newline
	\includegraphics[width=0.3\textwidth]{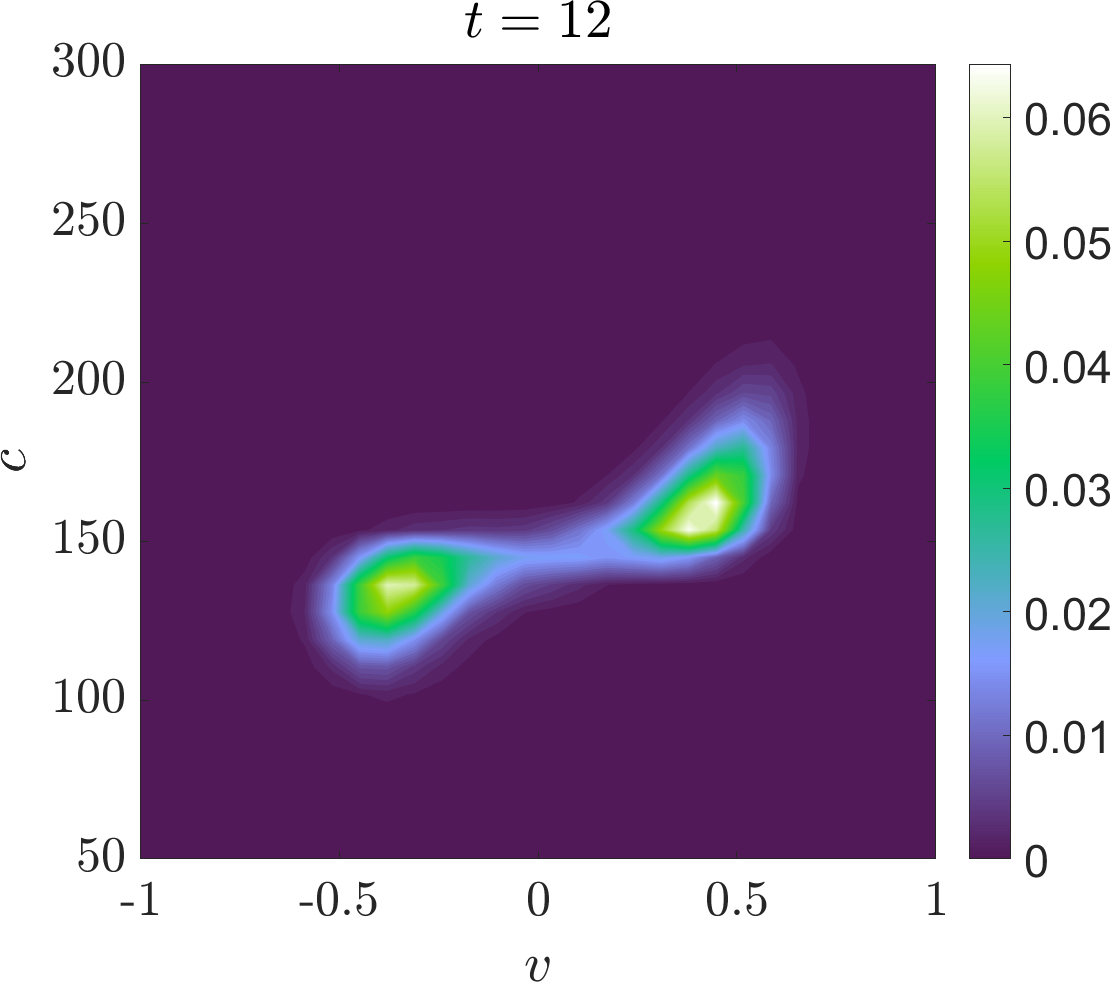} \quad\includegraphics[width=0.3\textwidth]{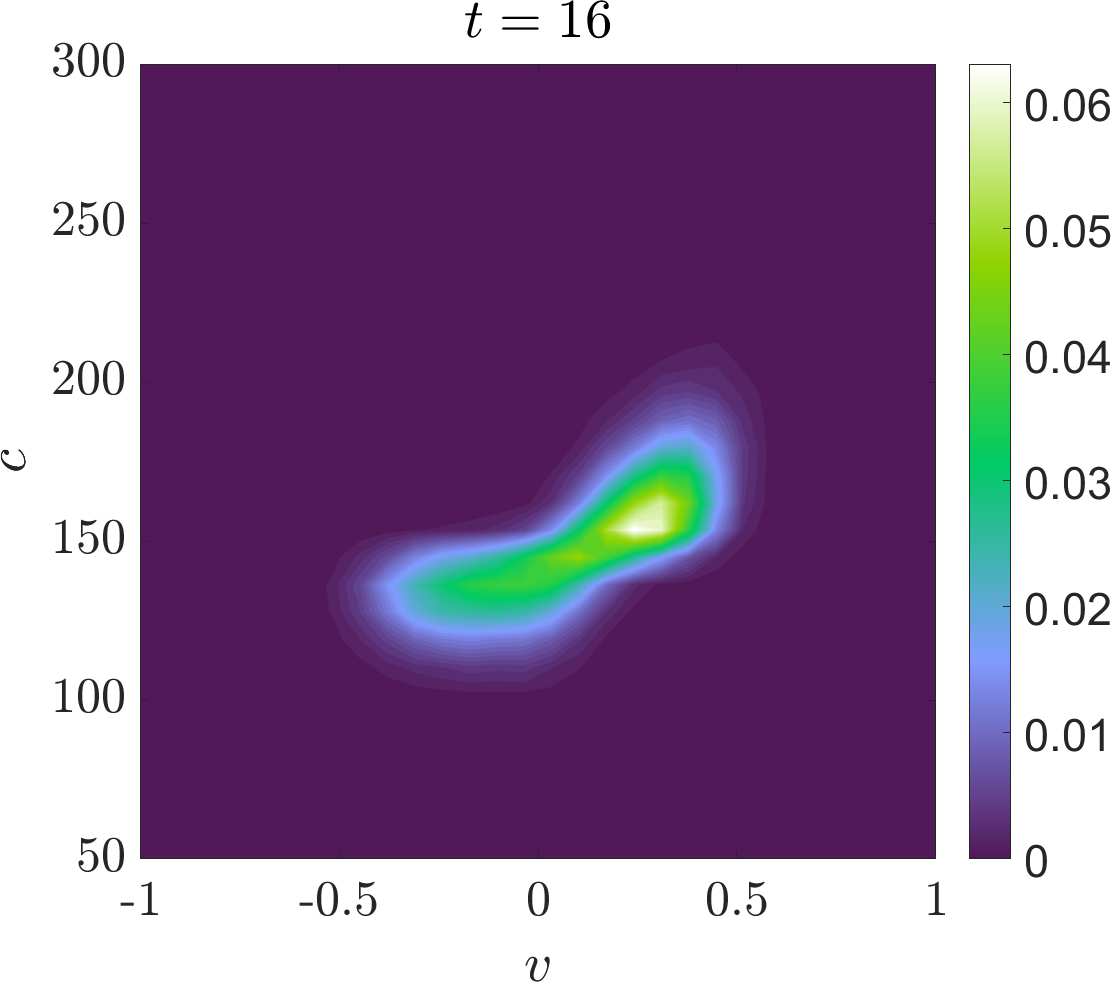} \quad\includegraphics[width=0.3\textwidth]{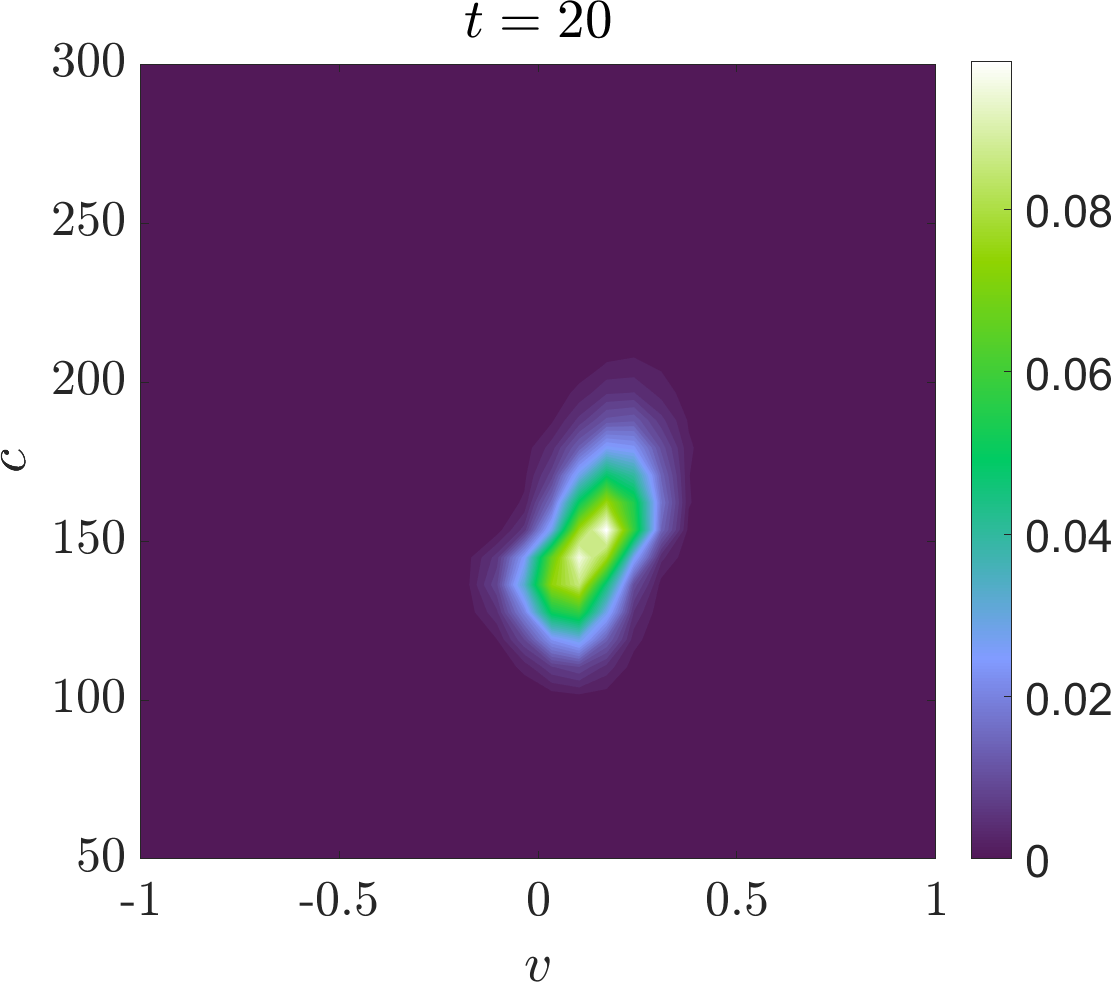} 
	\caption{Test $1$, $\sigma^2/\alpha = 0.005$. The pictures show the time evolution of the distribution function $f(v,c,t)$ for $t=0,4,8,12,16,20$ in the case of a non homogeneous distribution of the number of connections with respect to opinions. After the emergence of two clusters the agents reach consensus at the final time in the positive opinion region
	}\label{fig:test2bis2}
\end{figure}

\subsubsection{Test 2: Heterogeneous confidence bound}
In this second test, we modify the previous experimental setting introducing confidence bound $\Delta \equiv\Delta(c,c_*)$  in \eqref{eq:kernels} as a function of the contact numbers of the interacting agents. We assume that the confidence bound is above the threshold $1/2$ if the number of contacts of the interacting agent is larger, i.e. $c_*>c$, on the other hand, for a lower number of contacts $c_*<c$ the confidence bound becomes lower than $1/2$. To model this behavior we consider the following function
\begin{equation}\label{eq:hetconf}
	\Delta (c,c_*) = \frac{c_*}{c+c_*}.
\end{equation} 
We assume as initial datum $f^0(v,c)$ the uniform distribution on $ [-1,1] \times [0,1]$. Simulations of the opinion dynamics \eqref{kine-ww} is performed using $N_s= 10^5$ agents and scaling parameter $\epsilon = 0.01$. The evolution of social contacts \eqref{k1}, which initially is not at the stationary state, is performed with parameters $\mu = 10^{-1}$ and $\nu^2 = 0.0125$.

We compare the case with heterogeneous confidence bound as in \eqref{eq:hetconf},  with the case with homogeneous confidence bound equal to $\Delta(c,c_*)=1/2$. In Figure \ref{fig:test3bis} we report the resulting densities $f(v,c,t)$ for $t = 2$ on the left, $t=4$ in the center, and $t=8$ on the right, where the rows are relative to the homogeneous and to the heterogeneous case, respectively for top row and bottom row.
In the top row, we observe the emergence of two clusters one centered around $-0.5$ and the second one centered around $0.5$, independent of the number of contacts.
The bottom row reports the case with heterogeneous bound $\Delta(c,c_*)$,  where
we observe that agents with a high level of connections tend to maintain their opinions in time due to the lower values of $\Delta(c,c_*)$, instead agents with lower connections reach consensus around the central opinion $v=0$, since $\Delta(c,c_*)$ is larger.

\begin{figure}[h!]
	\centering
	\centering
	\includegraphics[width=0.3\textwidth]{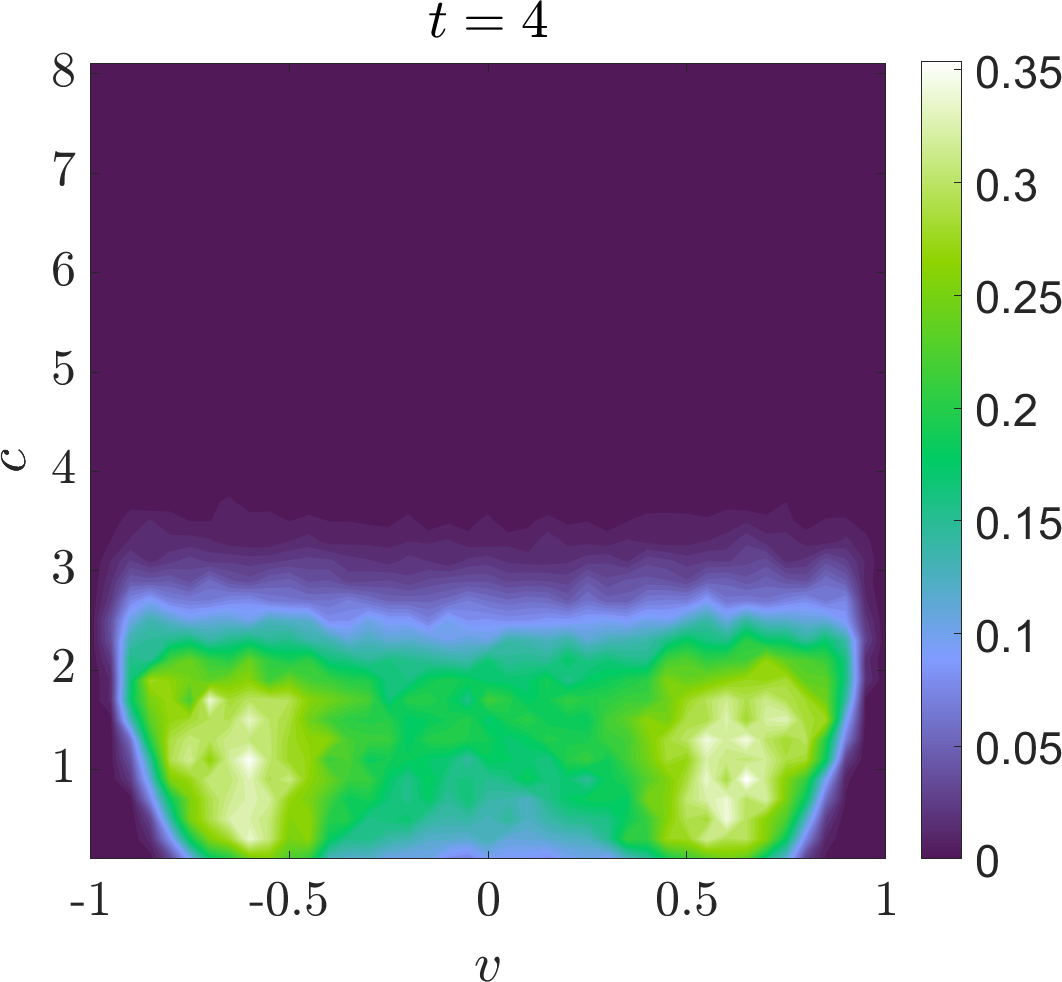} \quad\includegraphics[width=0.3\textwidth]{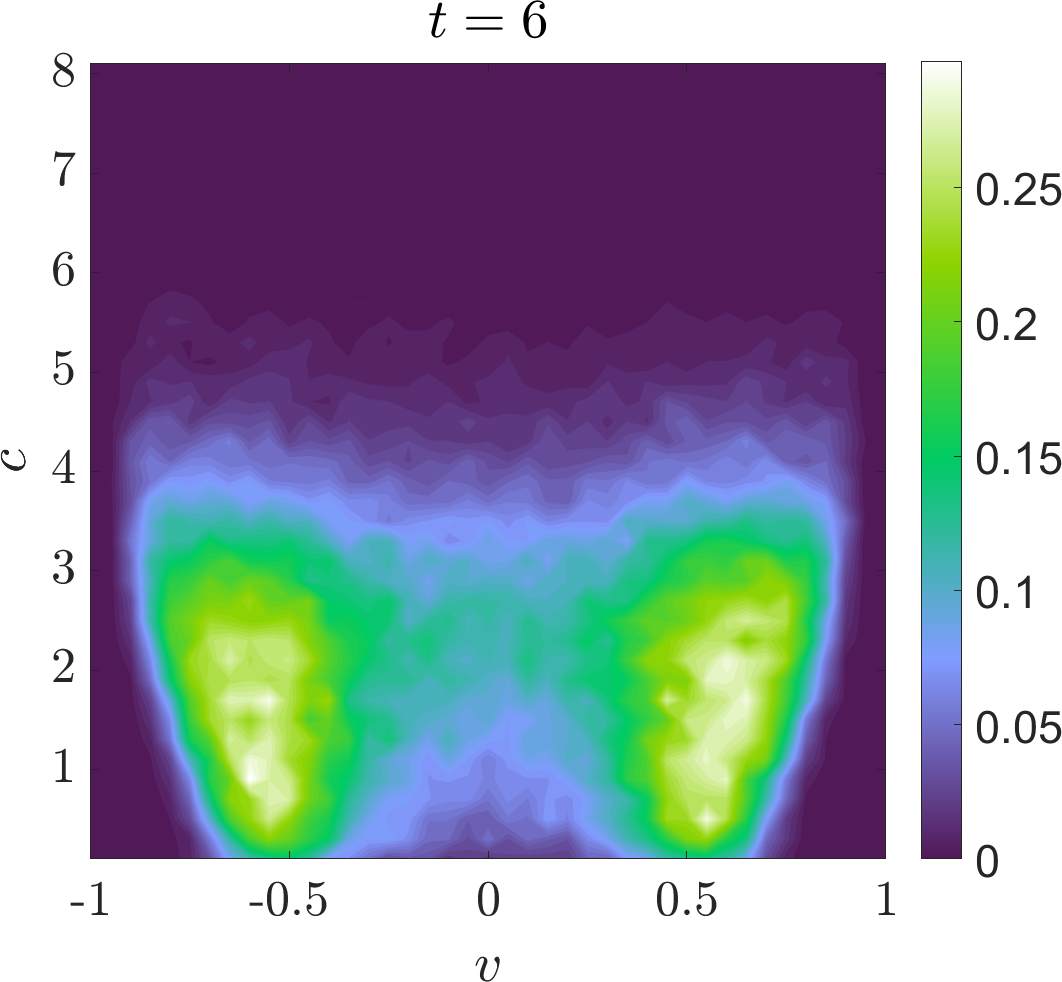} \quad \includegraphics[width=0.3\textwidth]{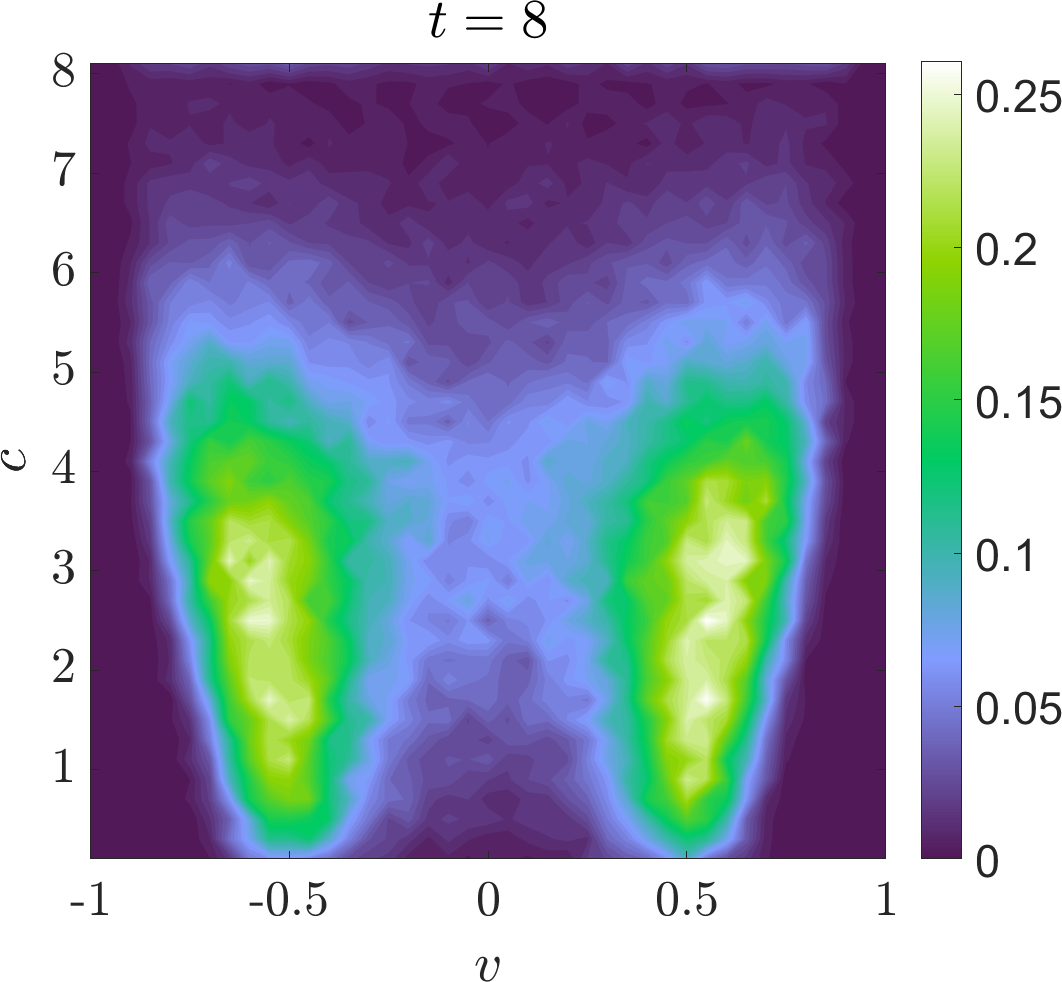} 
	\\
	\includegraphics[width=0.3\textwidth]{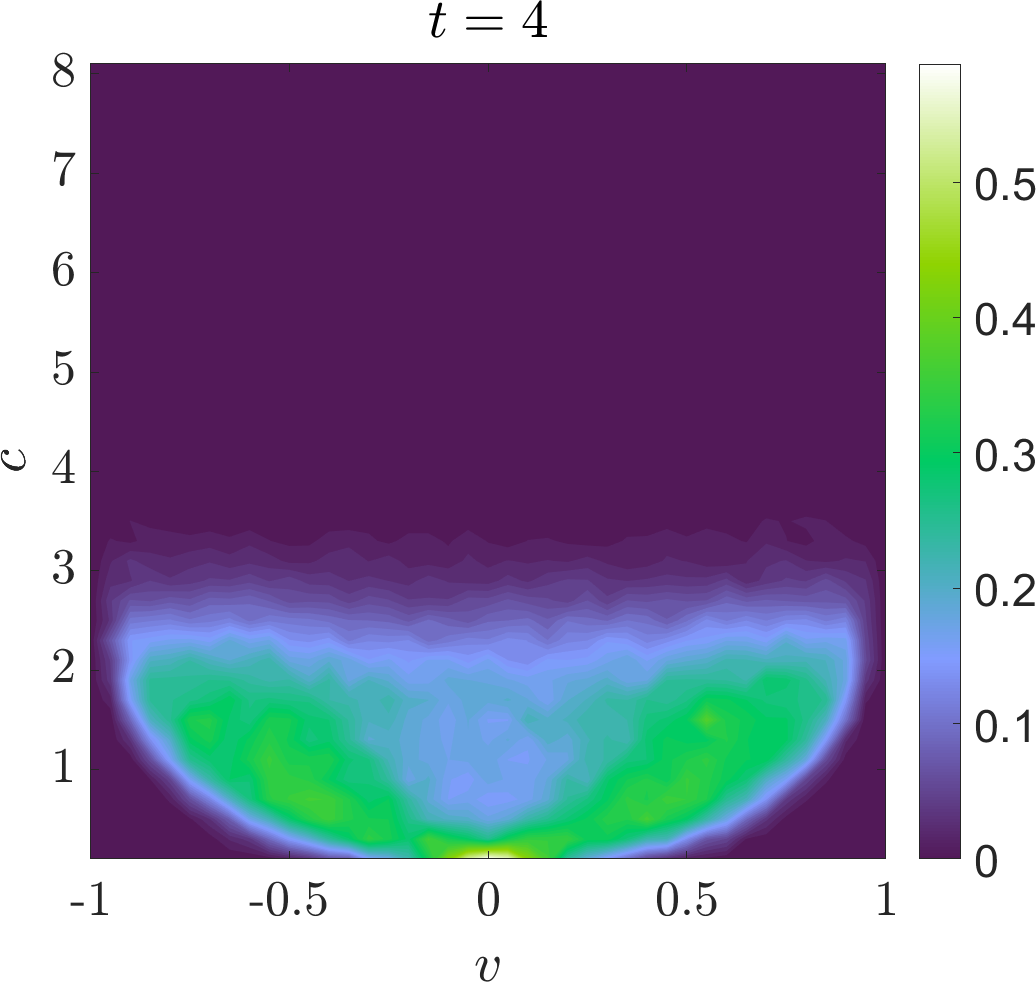} \quad\includegraphics[width=0.3\textwidth]{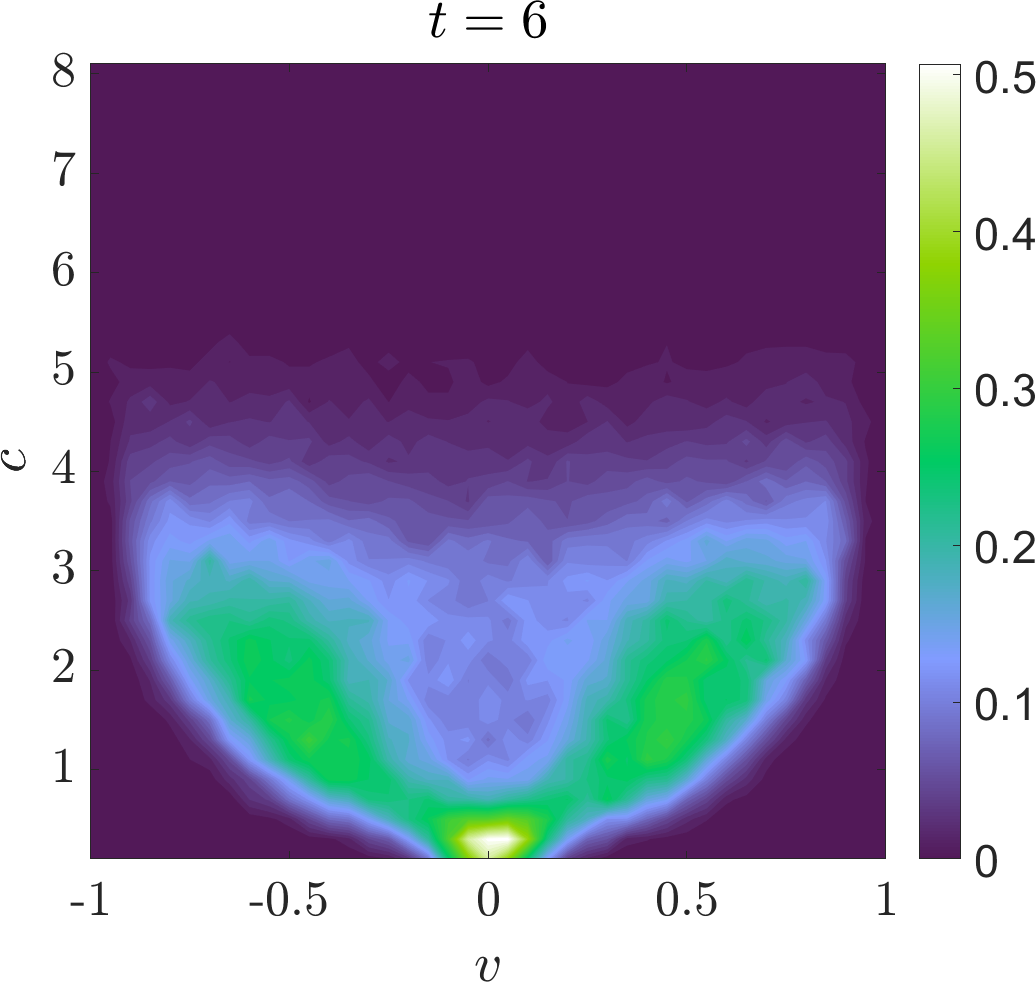} \quad \includegraphics[width=0.3\textwidth]{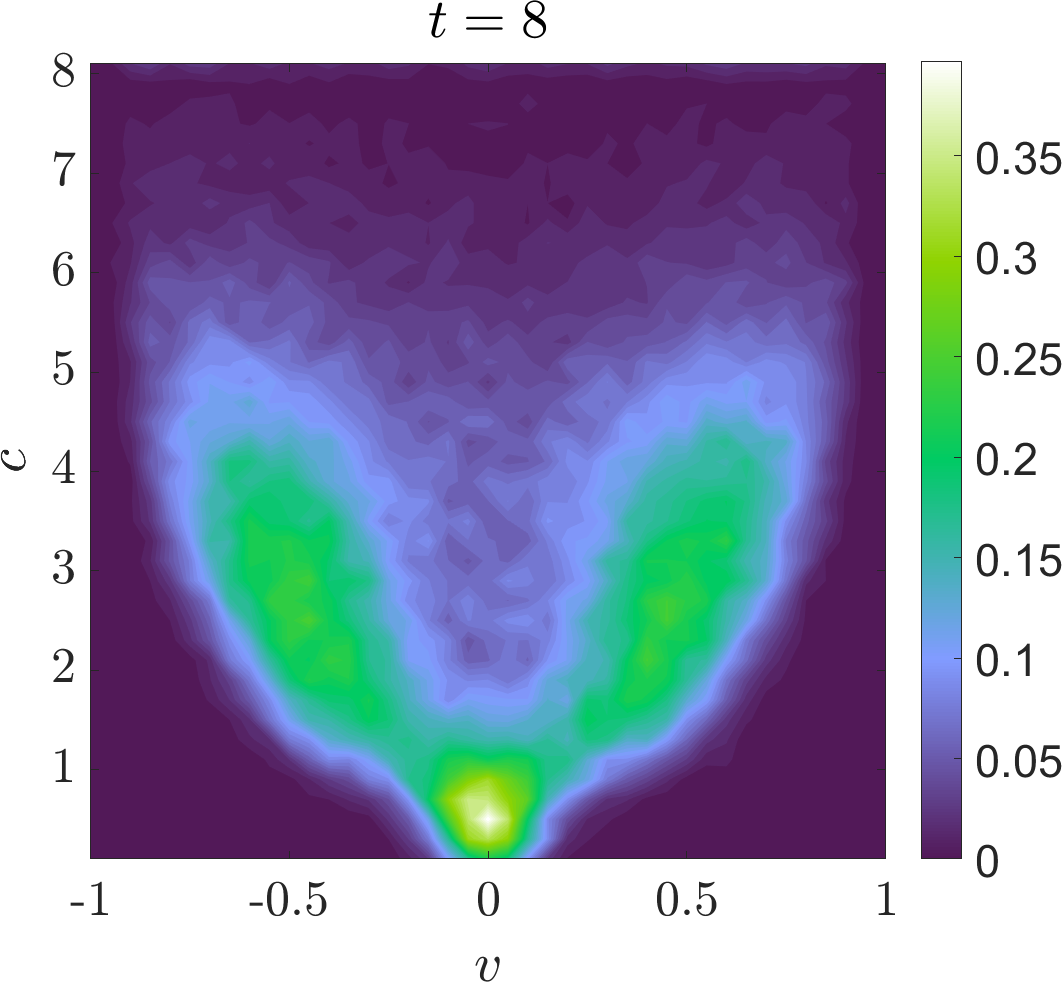} 
	\caption{Test $2$, $\sigma^2/\alpha = 0.005$. The pictures show the distribution $f(v,c,t)$ for $t=4$ (left), $t=6$ (center), and $t=8$ (right). Top row: constant bound of contacts ($\Delta(c,c_*)=0.5)$), two main clusters emerge at any bound of contacts.
		Bottom row: heterogeneous confidence bound ($\Delta(c,c_*)=$ in \eqref{eq:hetconf}), consensus is reached for agents with a low number of contacts, whereas for higher bound of contacts two main clusters emerge
	}\label{fig:test3bis}
\end{figure}


\subsubsection{Test 3: Sznajd-type dynamics}
We start now with a distribution of opinions given by the sum of two Gaussian distributions centered at different locations in $[-1,1]$ mimicking a clustering of opinion with respect to a given subject. The initial condition is
$$
f_0(v,c) = K_0\begin{cases}
	\frac{1}{30} \text{exp} \{-(v + \frac{3}{4})^2/(2\sigma_1^2)\},\quad& \text{if } 100 \leq c \leq 130 \\
	\frac{1}{30} \text{exp} \{-(v - \frac{3}{4})^2/(2\sigma_2^2)\},& \text{if } 170 \leq c \leq 200 \\
	0 &\text{otherwise,}\end{cases}
$$
with $K_0$ positive constant such that $f_0(v,c)$ has total mass equal to $1$. The interaction kernels are such that 
$$H(v,v_*,c,c_*) = (1-v^2), \qquad K(c,c_*) = \frac{c_*^2}{(c+c_*)^2},$$ and with a diffusion term proportional to $D(v,c,c_*) = 1-v^2$. The interaction function is chosen according to an approximation of Sznajd dynamics \cite{sznajd2000opinion,To06} in such a way that, individuals may interact with everyone on the social network while, alignment towards a given opinion, is more frequent for people having weak opinions and less probable for individuals having strong believes (positive or negative). The role played by the number of connections in the function $K(c,_*)$ is similar to the previous cases even if now its impact is more important: agents with a higher number of social media contacts tend to only slightly modify their opinion over time, while the ones with a low number of contacts are more influenced and tend to align their opinion with one of the most popular persons. The simulation uses $N_s = 10^5$ agents in the time interval $[0,T] = [0,6]$. Figure \ref{fig:test4bis} shows the initial condition on the left and the resulting density $f(v,c,t)$ for $t=3$ and $t=6$ respectively on the center and the right, using $\sigma_1^2=\sigma_2^2=0.005$ and a scaling parameter $\epsilon = 0.01$. For the evolution of the social contact dynamics \eqref{k1}, we account for the following parameters $\mu = 10^{-2}$ and $\nu = 3.54\cdot 10^{-2}$. We observe that agents with a lower number of contacts are strongly influenced by agents with a higher number of contacts, and behave like followers, whereas agents with higher-bound of contacts are mildly influenced by lower-contact agents. The result of such dynamics is that individuals with a negative opinion at the beginning are driven toward a positive one over time while people with a large number of contacts only slightly move toward the center without really modifying their believes.
%
\begin{figure}[h!]
	\centering
	\includegraphics[width=0.3\textwidth]{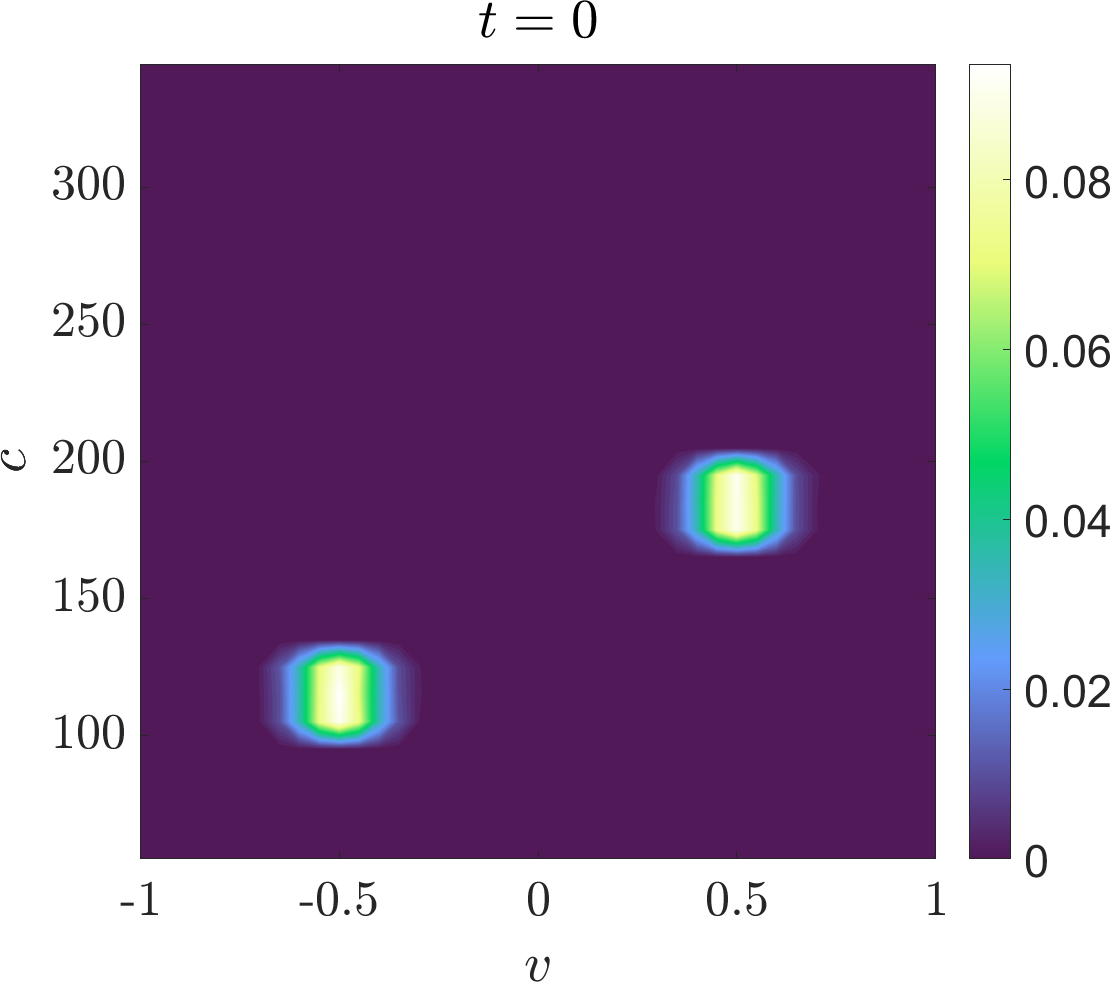} \quad\includegraphics[width=0.3\textwidth]{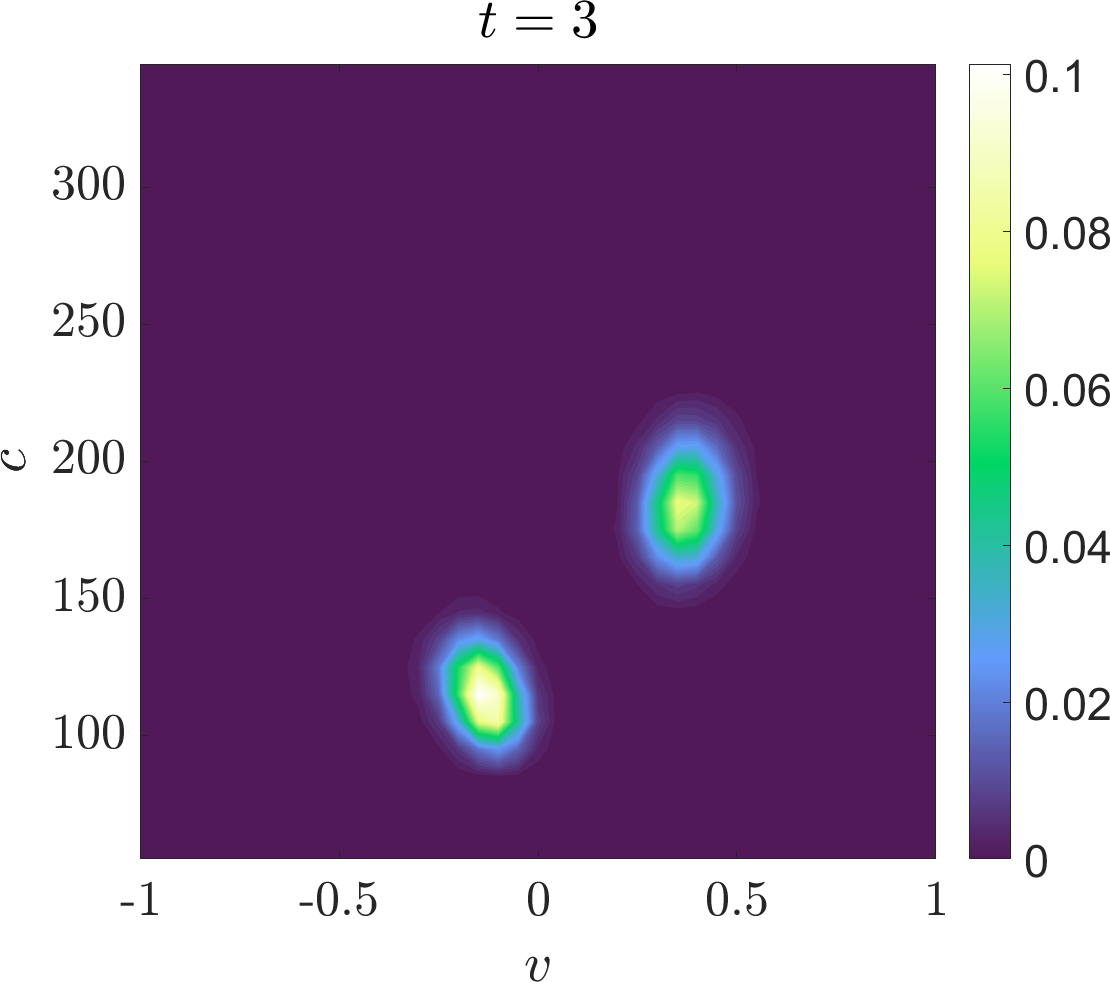} \quad\includegraphics[width=0.3\textwidth]{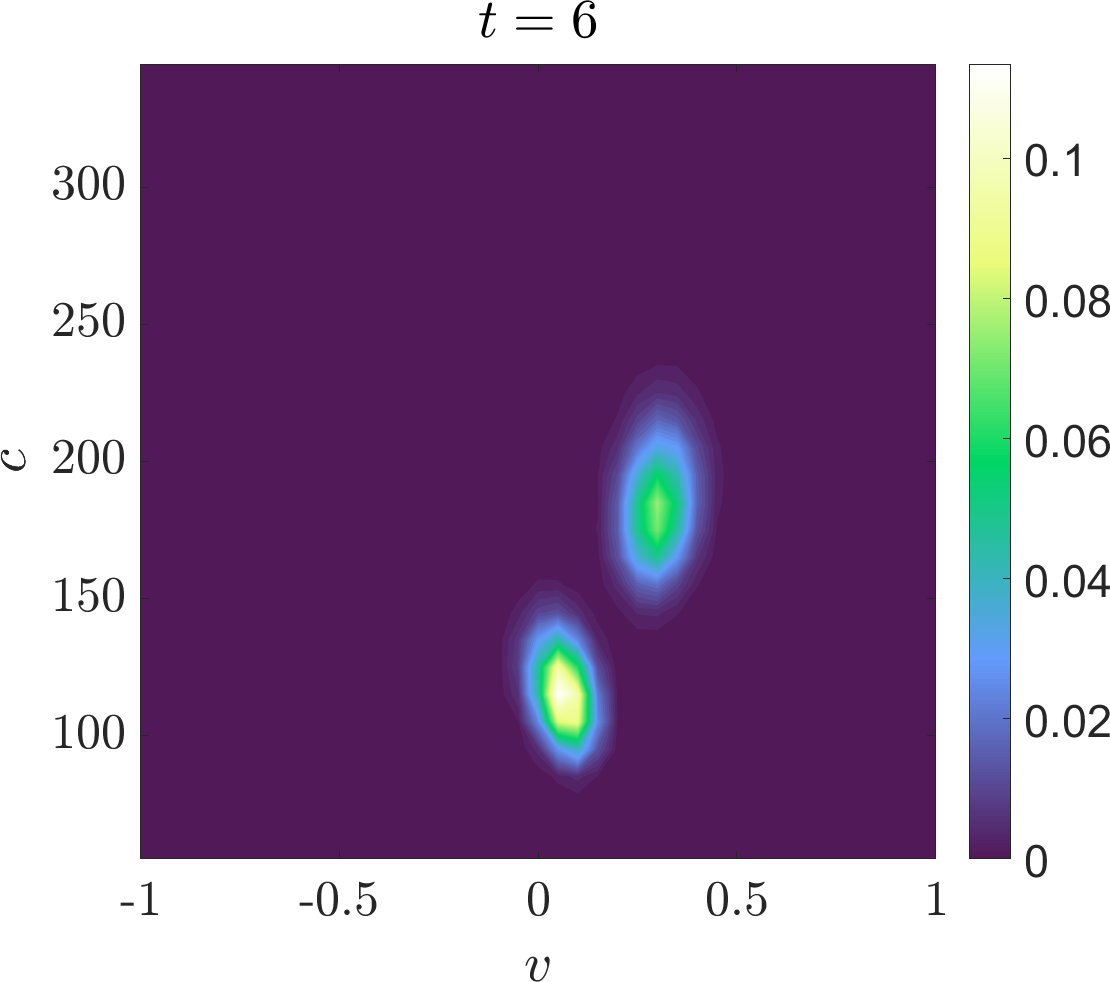} 
	\caption{Test $2$, $\sigma^2/\alpha = 0.005$. The pictures show the time evolution of the distribution function $f(v,c,t)$ for $t=0$ (left), $t=3$ (center), and $t=6$ (right). Agents with lower bound of connections are strongly influenced by agents with a large number of connections
	}\label{fig:test4bis}
\end{figure}

\subsection{Quantitative analysis}\label{qa}
In this last part, we use our model to quantitatively represent the opinion dynamics for what concerns some extrapolated real data taken from Twitter. To that aim, we start by describing the way in which the data are pre-processed through a so-called sentiment analysis with the scope of obtaining a set of information that can be effectively used in comparison with the model outcomes.

\subsubsection{Twitter sentiment analysis}
{\em Sentiment analysis} (or {\em opinion mining}) is a subfield of Natural Language Process (NLP) which works as a text classification tool that analyzes text data and extracts its intent, meaning if the underlying sentiment is positive, negative, or neutral. It is widely used for commercial purposes, in order to monitor brand and product sentiment in customer feedback and understand customer needs. Sentiment analysis methods can be divided into two main categories: statistical methods based on machine learning algorithms, and knowledge-based methods. Knowledge-based sentiment analysis approaches rely on a list of words called {\em sentiment lexicon} labeled as positive or negative, but they typically do not include sentiment-bearing lexical items as acronyms, emoticons, or slang terms (which are widely used in social media texts). Moreover, they do not account for differences in the sentiment intensity of words. Since many applications would benefit from determining not only the binary polarity but also the strength of the sentiment, some {\em sentiment intensity lexicons} have been implemented in the past, that associate a sentiment valence to the words and help measuring the intensity expressed in the sentences. This sentiment valence translates into a value called {\em polarity} which ranges between $-1$ and $1$ which is very well adapted to be employed in synergy with our model. One of the main problems of the natural language classifiers is that manually creating and validating a comprehensive sentiment lexicon is time intensive, so typically some machine learning approaches are incorporated to improve efficiency. These approaches come also with some drawbacks since they require intensive training data which are usually hard to acquire, they rely heavily on the vastness of the training set, they remain expensive in terms of CPU use, and they often make use of ``black boxes'' not easily interpretable and, therefore, not easily modifiable or generalizable.

In our analysis, we use VADER (Valence Aware Dictionary for sEntiment Reasoning), presented by C.J. Hutto and E. Gilbert in \cite{Hutto} in 2014. It uses a combination of qualitative and quantitative methods to build a list of lexical features that allow to perform sentiment analysis. This engine is specially constructed to give reliable results on social media texts and does not require a training dataset, since it is developed on a valence-based human curated lexicon. VADER's lexicon incorporates preexisting well-established word banks and several lexical features typical of social media, such as emoticons, acronyms, initialisms, and commonly used slang terms that are sentiment-related. These features were rated on a scale from $-4$ ``extremely negative'' to $+4$ ``extremely positive'' (normalized to $[-1;+1]$ in Python) using a ``wisdom-of-the-crowd'' approach, meaning that the ratings were given starting from the collection of answers of a series of independent human raters. VADER's developers then used a deep qualitative analysis resulting in isolating five generalized heuristics based on grammatical and syntactic cues to determine the sentiment intensity of short sentences. 
In the next sections \ref{test4} and \ref{test5} we use VADER to perform sentiment analysis on two different data sets of actual tweets.

\subsubsection{Test 4: Trump re-admission on Twitter}\label{test4}
For this simulation, we used the Application Programming Interfaces (API) of Twitter to obtain the content of a certain number of tweets on a given topic. More specifically we used the words ``Donald Trump'' and some related hashtags to select tweets written (in English) some days after the re-admittance of Donald Trump on Twitter on the 20th of November 2022, after an almost two-year long ban from the platform. We then employed VADER to analyze the texts of the tweets and we obtained a rating between $-1$ and $1$ for each tweet, which we considered to be the agents' opinions on the subject. To outline the polarized situation we remove tweets, which have scored exactly $``0"$ through VADER analysis.

To perform the model calibration we introduce a class of interacting functions,  where we make explicit the dependency with respect to a new set of parameters $\theta \in \Theta \subseteq \mathbb{R}^4_+$  as follows
\begin{subequations}\label{eq:subkern}
	\begin{equation}\label{eq:kernel_par}P(v,v_*,c,c_*;\theta) = H(v,v_*,c,c_*;\theta)K(c,c_*;\theta),
	\end{equation} 
	with 
	\begin{equation}\label{eq:H}
		H(v,v_*,c,c_*;\theta) = \chi(|v-v_*|< \Delta(c,c_*;\theta)),
	\end{equation} 
	where
	\begin{equation}\label{eq:delta}
		\Delta(c,c_*;\theta) =  \theta_1\left(\frac{ \log(1+c_*)}{\log(1+c_*)+\log(1+c)}\right)^{\theta_2},
	\end{equation} 
	and
	\begin{equation}\label{eq:K}
		K(c,c_*;\theta) = \left(\frac{ \log(1+c_*)}{\log(1+c_*)+\log(1+c)}\right)^{\theta_3},
	\end{equation} 
	while the diffusion is weighted by
	\begin{equation}
		D(v,c;\theta)  = \theta_4\sqrt{1-|v|^2}.
	\end{equation} 
\end{subequations}
The initial data is well prepared, assuming that the distribution of connections is at the stationary state \eqref{eq:logn} with parameters estimated in \ref{tab:t1}, and the joint distribution of opinions and contacts is such that
\[
f_0(v,c) = \begin{cases} \frac{h_\infty(c)}{2}, \quad &\mbox{if } 50 < c \leq 7500\\
	\frac{h_\infty(c)}{0.2\sqrt{2 \pi}}e^{-\frac{1}{2}\left(\frac{v + 0.5}{0.2} \right)^2}, &\mbox{if } c > 7500.
\end{cases}
\]
Finally,  we identify the distribution obtained from the data with $\hat g(v,t)$ while $g(v,t)$ is the one obtained simulating the virtual dynamics of particles. Then, we search the optimal value of the parameters $(\theta_1,\theta_2,\theta_3,\theta_4)\in\Theta$, where
$ \Theta =[0.5,1.5]\times[0.2,0.7]\times[1.5,2.5]\times[8,11],$
by minimizing the $\ell_1$ distance at final time $T= \{20\}$ of the marginal distribution of the simulated opinions $g(v,T)$ and the one reconstructed from data $\hat g(v,T)$.

	The minimization is performed using \texttt{patternsearch()} routine in \texttt{matlab} with initial guess $\theta^{(0)}=(1.1,0.4, 2.0, 10)$ and reaching convergence after  $k =60 $ iterations, where the estimated parameters is $\theta^{(k)}=(0.6313,0.2047, 2.3125; 9.7510)$. Here we minimize the discrepancy measure as the $\ell_1$ distance between the marginal distribution $g(\cdot,T|\theta^{(k)})$ and the marginal distribution reconstructed from data $\hat g(\cdot,T)$,  the final value is $\mathcal{D}_1(g(\cdot,T|\theta^{(k)}),\hat g(\cdot,T)) = 6.376\times10^{-2}$.

The minimization procedure is further detailed in the Appendix.
		%

In Figure \ref{fig:test4} we depict the comparison between our results and the data: on the left, we have the marginal of the opinions at time $t=20$ and the data, while on the right we have $f(v,c,t)$ (represented as $\log(f(v,c,t)+0.025)$) compared with the actual data extracted from Twitter (represented by the orange dots). We can claim that the model is capable of fitting the results obtained from a sentiment analysis with good accuracy.
\begin{figure}[h!]
	\centering
	\includegraphics[height = 5.85cm]{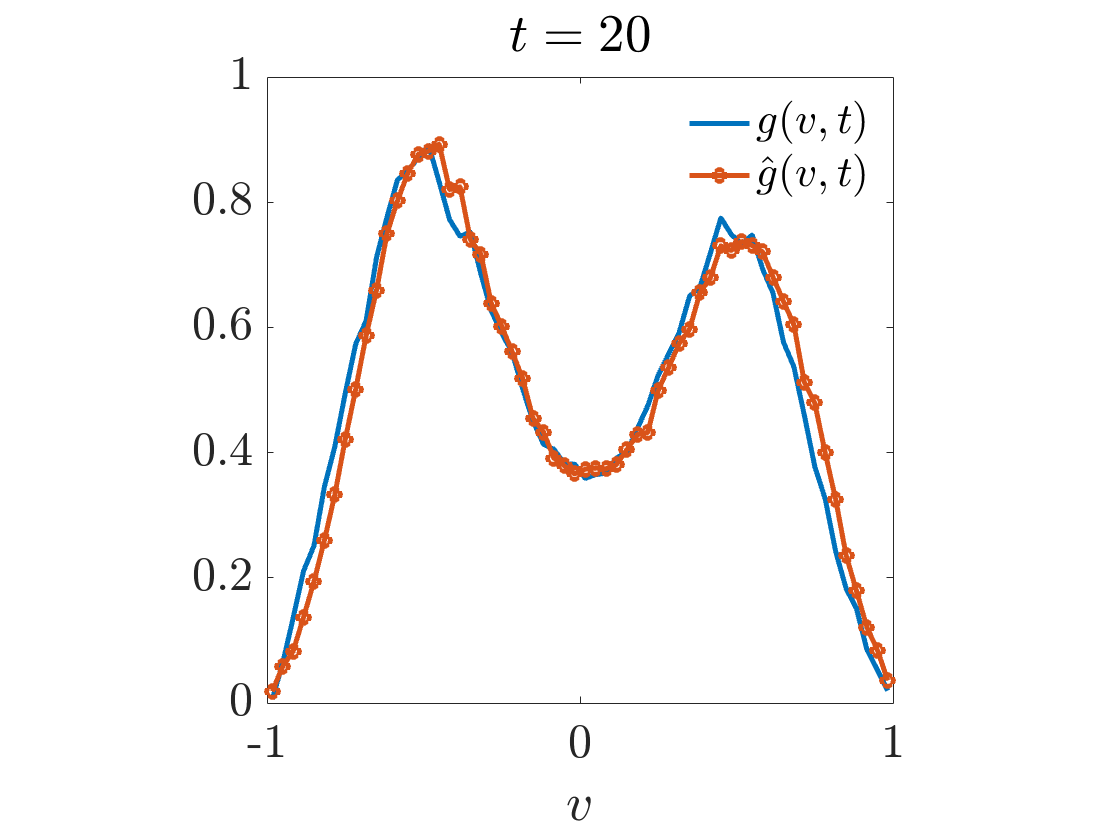}
	\includegraphics[height = 5.8cm]{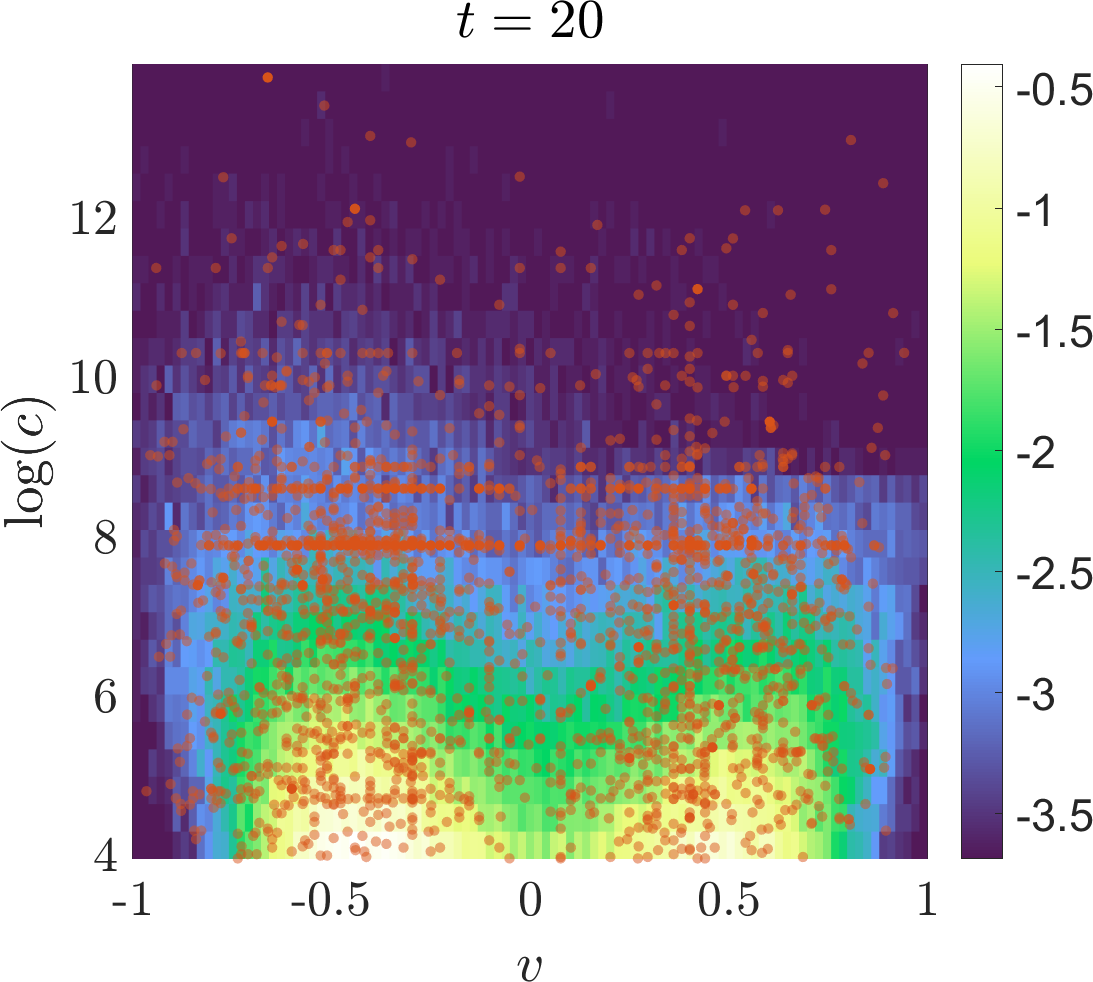}
	\caption{Test $4$: marginal distribution of opinions at the final time (left) and the comparison between the reconstructed density of opinions and contacts and the real data set}\label{fig:test4}
\end{figure}

\subsubsection{Test 5: Climate change trends on Twitter}\label{test5}
For this last simulation, we used the information coming from a {preexisting dataset (\cite{dataset_climatechange}) } containing the IDs of tweets discussing the \emph{climate change}. Such data were collected from Twitter's API between September \nth{21}, 2017 and May \nth{17}, 2019, using as track parameters some keywords related to the subject, such as ``climate change", ``global warming", and hashtags like ``$\#$climatechangeisreal", ``$\#$climatechangeisfalse", or ``$\#$globalwarminghoax". 
We used VADER to perform the sentiment analysis on the tweets collected on August \nth{13}, \nth{14}, \nth{15}, \nth{16}, and \nth{18} of the year 2018, and we remove from the data-set tweets with SA score exactly equal to $``0"$,  to reconstruct the data trend starting from a fictitious initial distribution of the opinions. 

Similarly to the previous Test \ref{test4} we consider the interaction functions reported in \eqref{eq:subkern},  and with well-prepared initial data, where we assume that the marginal distribution of the contacts is at the stationary state \eqref{eq:logn}
, and that the joint initial density of opinions and contacts is given by 
$$
f_0(v,c) = \begin{cases}\frac{h_\infty(c)}{154 \sqrt{2\pi}}\left(55\sqrt{2\pi} + 200e^{-\frac{1}{2}\left(\frac{v + 0.35}{0.15} \right)^2} + 28e^{-\frac{1}{2}\left(\frac{v - 0.25}{0.5} \right)^2}\right), \quad &\mbox{if }  c < 40\\
	\frac{h_\infty(c)}{6\sqrt{2\pi}}\left(250e^{-\frac{1}{2}\left(\frac{v + 0.8}{0.004} \right)^2} + 20e^{-\frac{1}{2}\left(\frac{v + 0.3}{0.2} \right)^2} + 5e^{-\frac{1}{2}\left(\frac{v - 0.3}{0.2} \right)^2} \right), &\mbox{if } 40 \leq c \leq 400, \\
	\frac{5 h_\infty(c)}{\sqrt{2 \pi}}e^{-\frac{1}{2}\left(\frac{v - 0.4}{0.2} \right)^2}, &\mbox{if } c > 400.
\end{cases}
$$
We assume that the data are referred to the following numerical time
$t_m = \{1, 2, 3, 4, 11\}$ and we denote by $\hat g(v,t)$ the empirical marginal distribution of the opinions obtained from Twitter, meaning that $\hat g(v,t_1)$ refers to August \nth{13}, $\hat g(v,t_2)$ refers to August \nth{14} and so on. 

Hence, to obtain the optimal value of the parameters $(\theta_1,\theta_2,\theta_3,\theta_4)\in\Theta$ in the admissible space
	\[
	\Theta = [0.5, 1]\times[0.1, 1.5]\times[0.01, 2]\times[0,0.05].
	\]
	we use the \texttt{fmincon()} \texttt{matlab} routine to minimize the discrepancy between data-reconstructed and simulated marginal distribution of opinions computing the sum over $t_m, m = 1,\dots,5,$ of the $1-$Wasserstein distances $\mathcal{W}^1_1(g(\cdot,t_m|\theta),\hat g(\cdot,t_m))$ as follows
	$$\mathcal{D}_1(g(\cdot|\theta),\hat g(\cdot)) = \frac{1}{M}\sum_{m=1}^M\mathcal{W}^1_1(g(\cdot,t_i|\theta),\hat g(\cdot,t_i)), \qquad M=5.$$
The chosen method reaches convergence after $k=11$ iterations with initial guess $\theta^{(0)}=(0.75,1.25,0.65,0.03)$ and
the estimated values of the parameters are $\theta^{(k)} = (0.7432,1.0735,0.9295,0.0306)$, resulting in a discrepancy of value $\mathcal{D}_1(g(\cdot|\theta),\hat g(\cdot)) =4.893\times10^{-1}$. We outline that the choice of a good initial guess, in this case, is of paramount importance.

In Figure \ref{fig:test5} we depict the evolution of the marginal distribution of opinions compared to the data, while Figure \ref{fig:test5_2} shows the initial and terminal density of opinions and contacts.
Again we can claim that we are able to follow the main trend of the opinion during the time even if compared to the previous situation of Test 4 in which the fitting was only about a given, supposed steady, state, here the differences are quite large in some time frameworks.

\begin{figure}[h!]
	\includegraphics[width=0.3\textwidth]{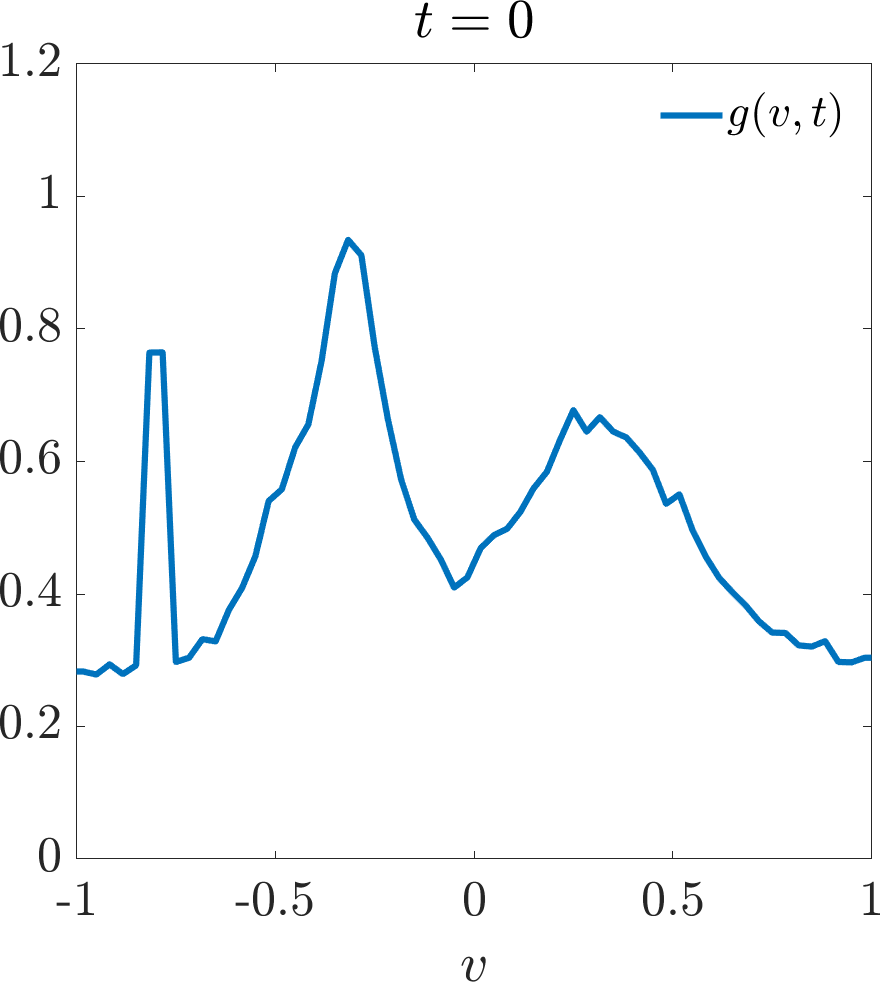} \quad\includegraphics[width=0.3\textwidth]{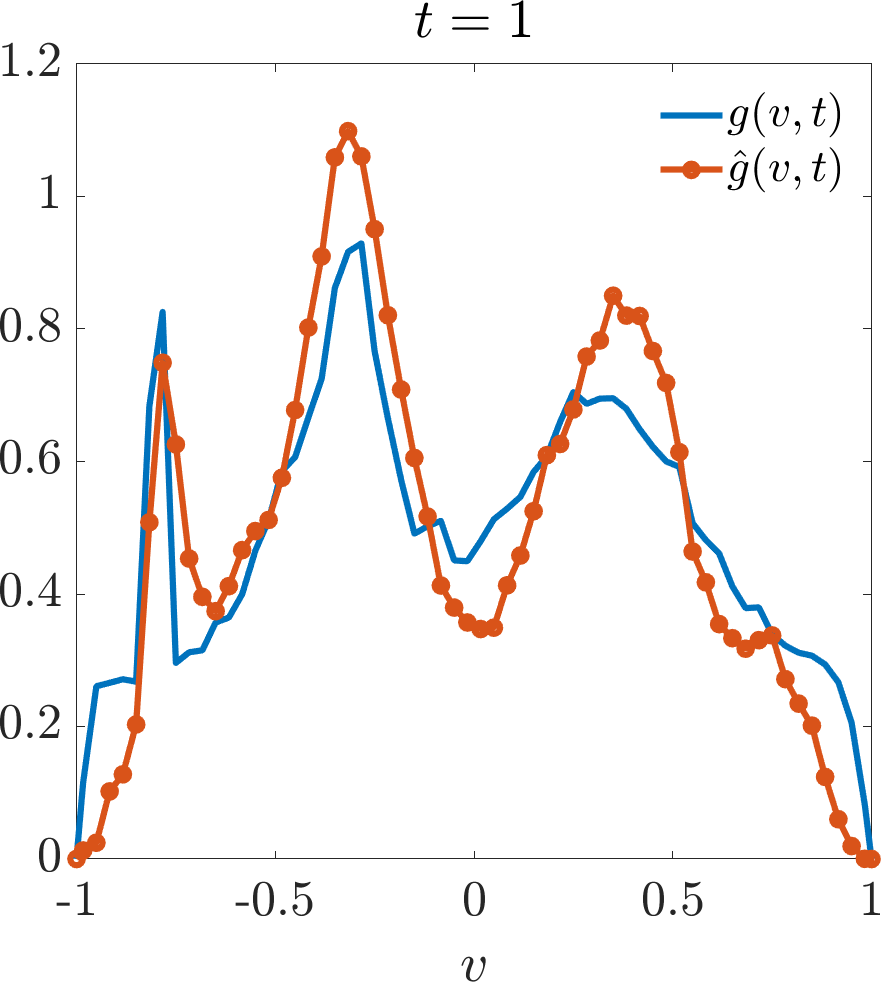} \quad\includegraphics[width=0.3\textwidth]{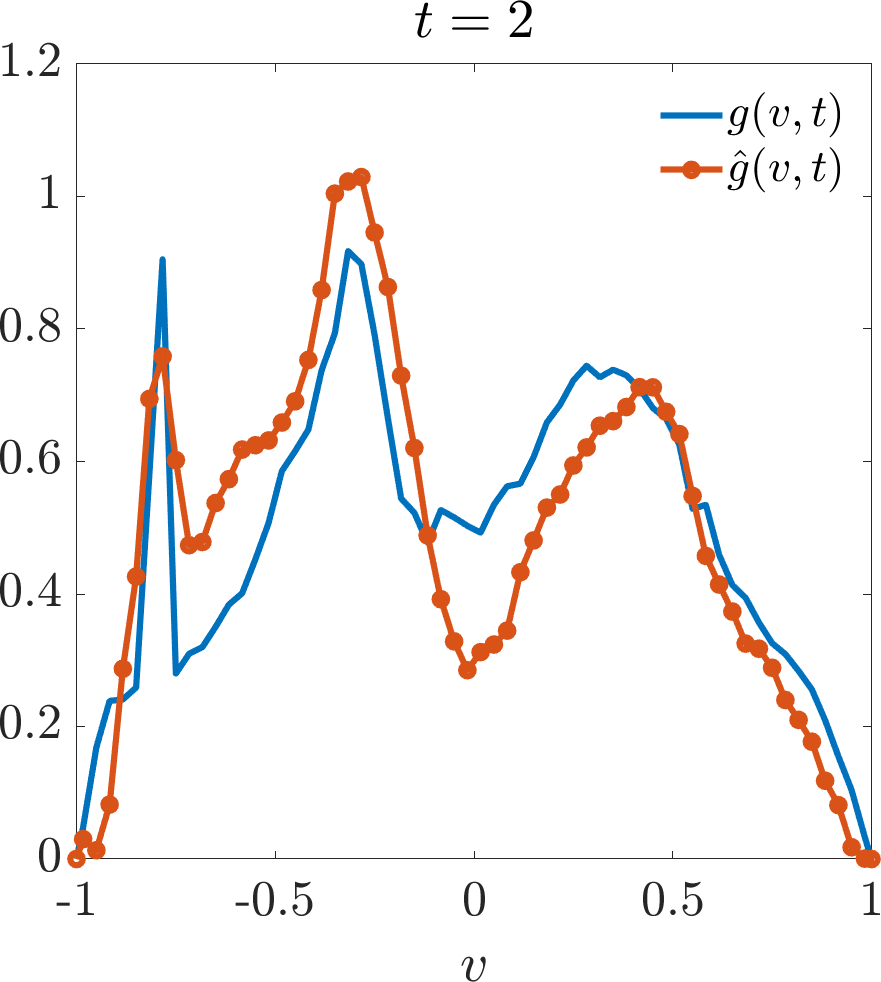}
	\newline
	\newline
	\includegraphics[width=0.3\textwidth]{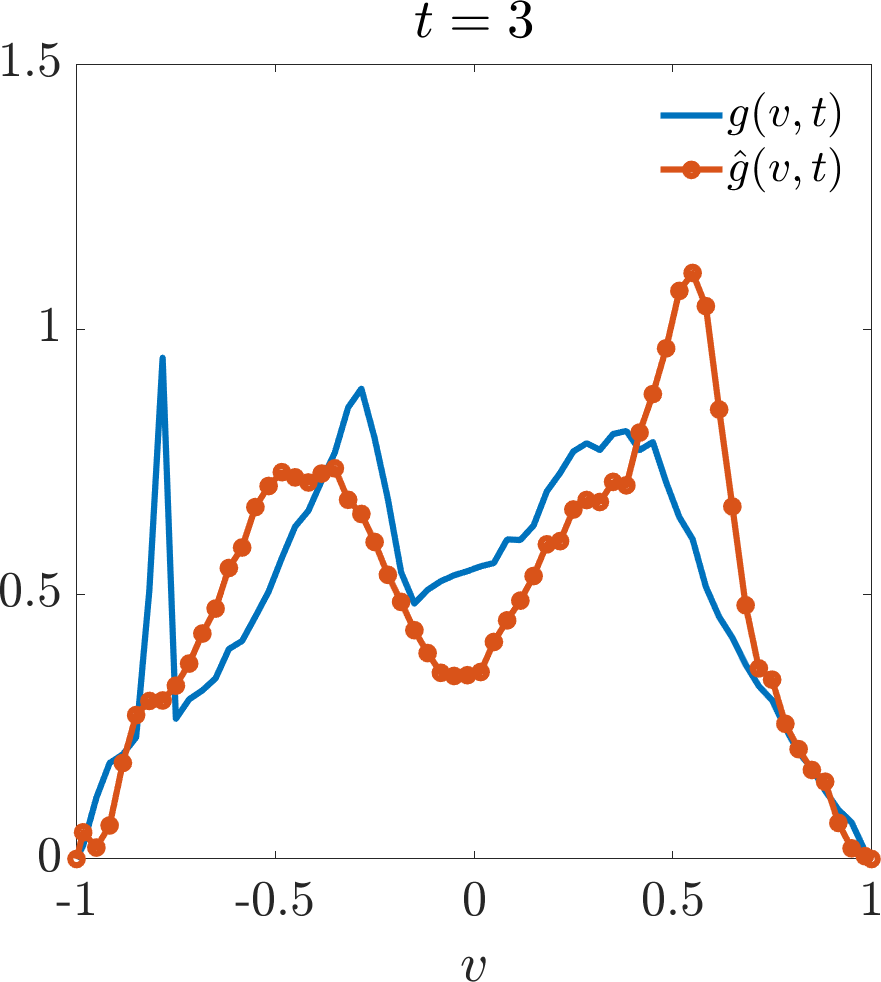} \quad\includegraphics[width=0.3\textwidth]{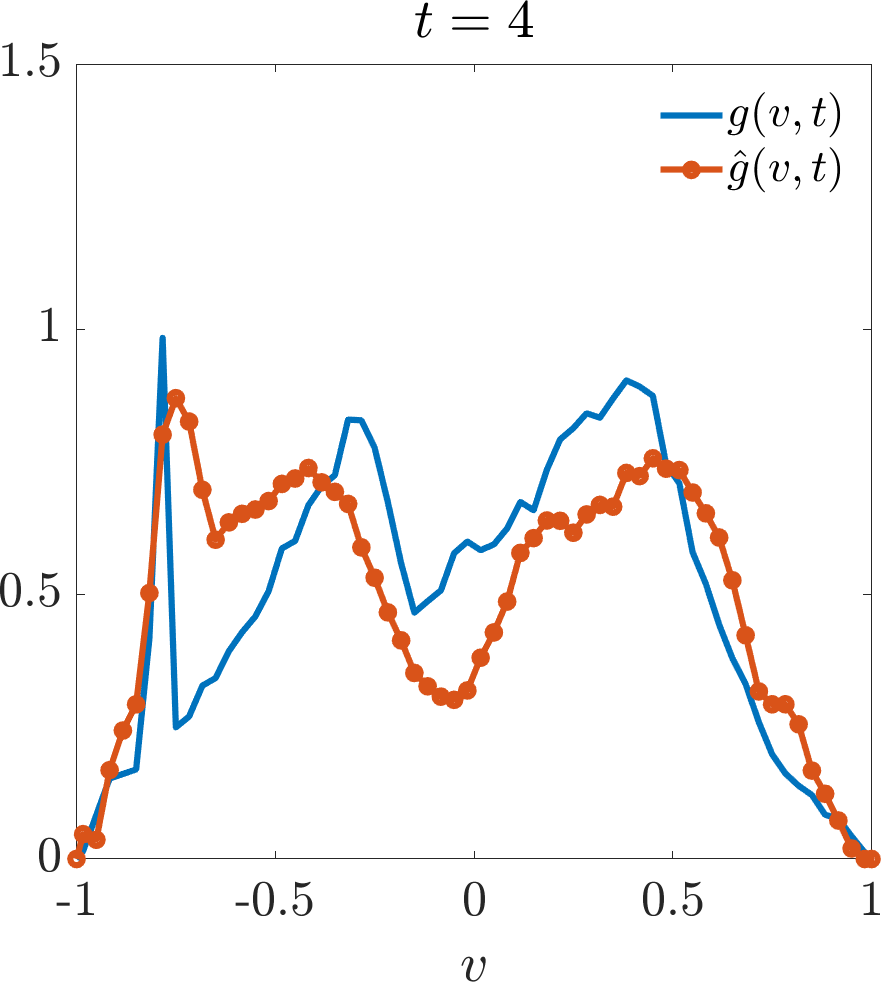} \quad\includegraphics[width=0.3\textwidth]{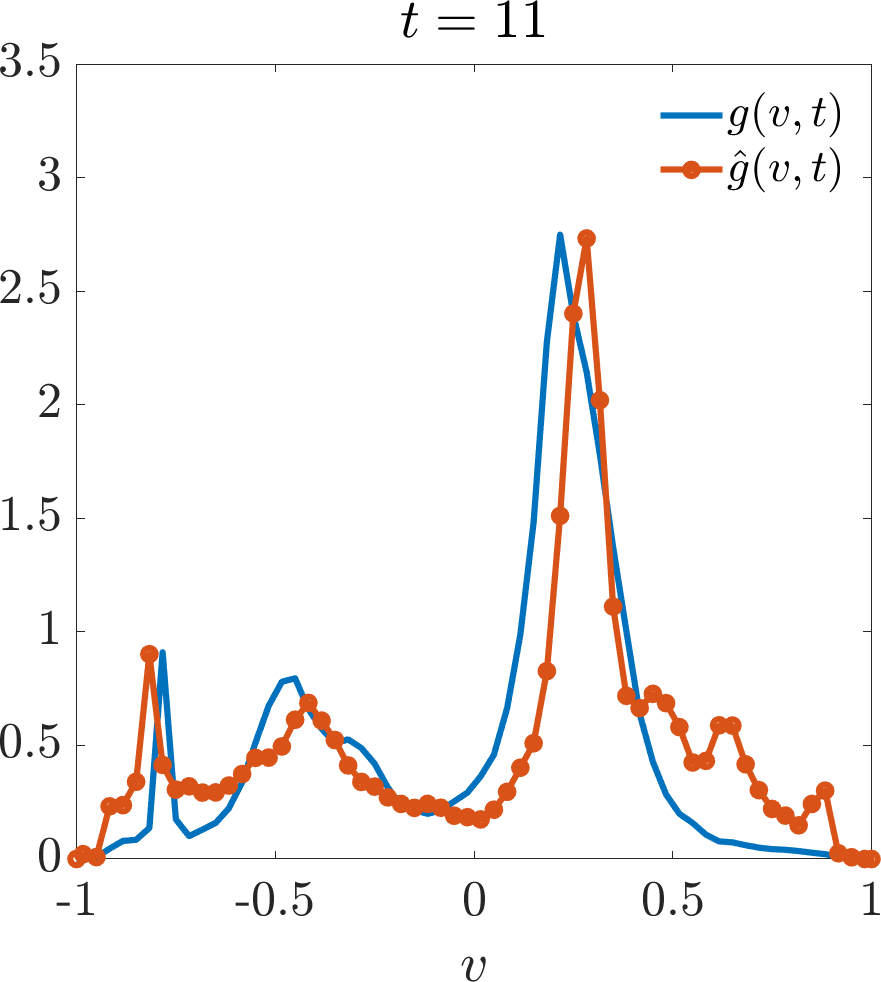} 
	\caption{Test 5: comparison between the marginal distribution of the opinions reconstructed using the presented model and the data at each time step $t \in \{1,2,3,4,11\}$ corresponding to data relative to the $13^{th}, 14^{th},15^{th},16^{th}$ and $18^{th}$ of August 2018}\label{fig:test5}
\end{figure}

\begin{figure}[h!]
	\includegraphics[width=0.5175\textwidth]{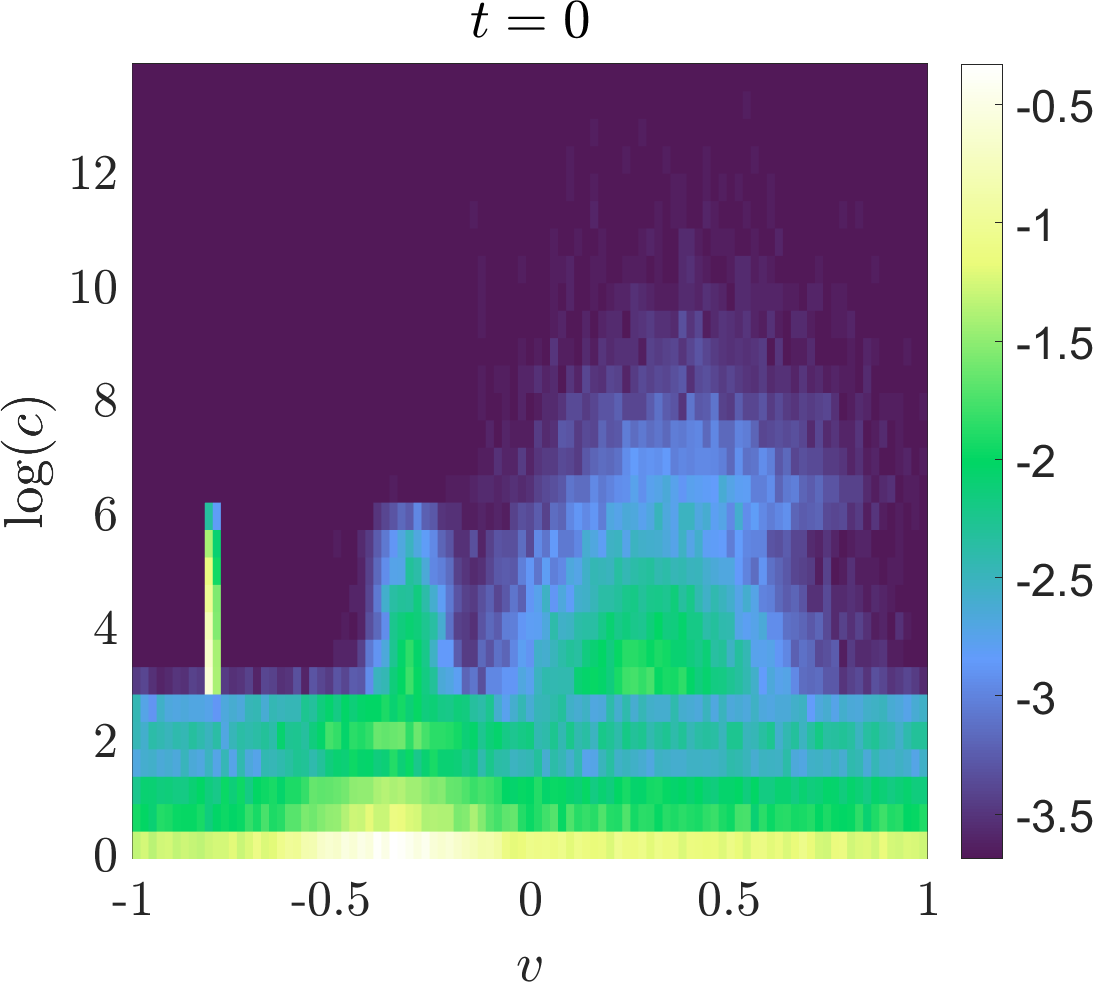}\quad\includegraphics[width=0.5\textwidth]{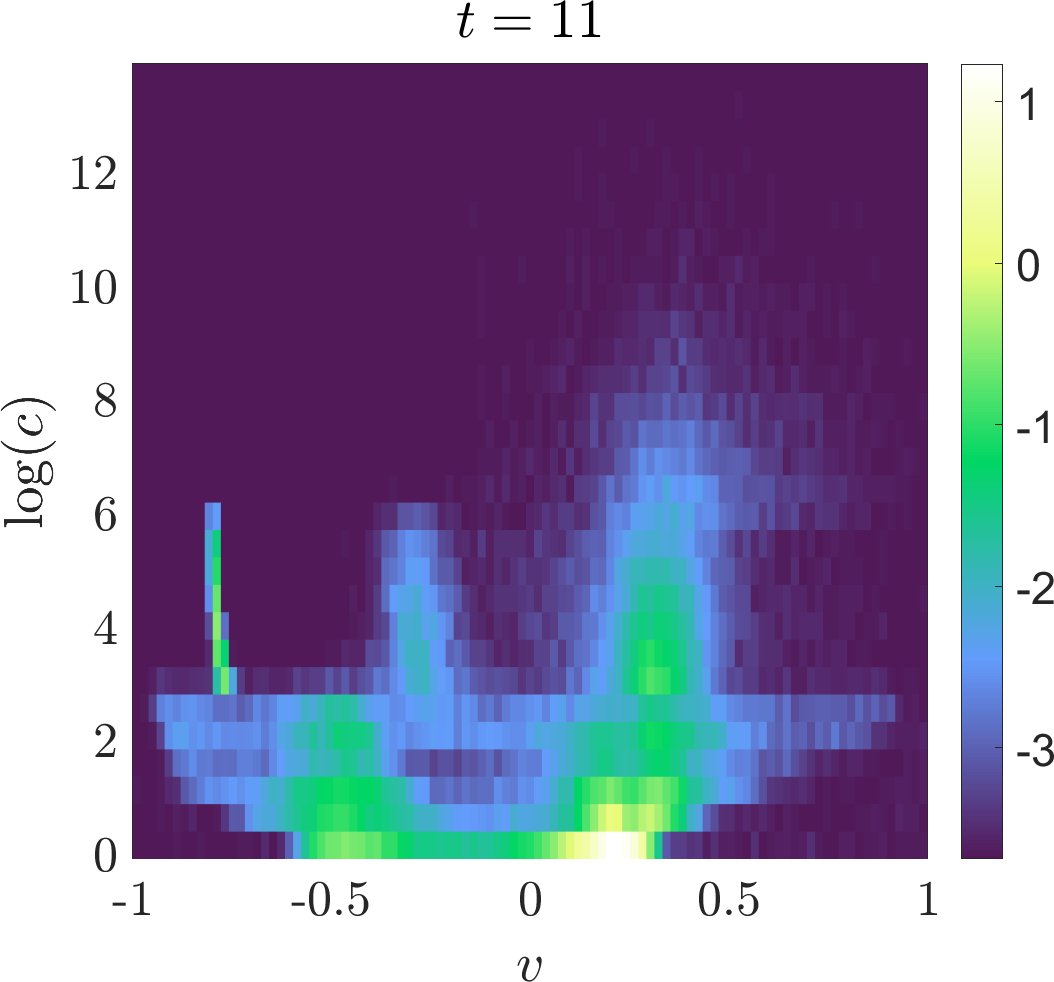}  
	\caption{Test $5$: initial (left) and final (right) joint density of opinions and contacts (the images show $\log(f(v,c,t) + 0.025)$) }\label{fig:test5_2}
\end{figure}

\section{Conclusions}\label{sec:conc}
In this work we have proposed a new model for opinion formation and evolution in presence of social media connections. Starting from a set of microscopic interactions characterizing the behavior of individuals acting on a social media platform, we first constructed a model for the network formation and then we tested the validity of our hypothesis through a comparison with a real data set extrapolated from Twitter. A fitting procedure permitted to recover the best set of parameters, which employed in our model is able to describe a network of people exchanging ideas and information over the net.
In the second part, we concentrated on the relationship between opinion and connections, and through a Boltzmann-like approach we recovered an equation 
which can describe opinion dynamics over a social network. Through a grazing limit procedure, we then obtained a Fokker-Planck asymptotic limit equation from which the main features of the model can be more clearly understood.
In the third part, we first performed some numerical simulations with the scope of showing some of the qualitative features of this new model and finally
thanks to sentiment analysis tools which permitted us to obtain realistic distributions of opinions on a certain topic from a given social media platform, we have shown that our model is indeed able to describe such dynamics. The results have been obtained thanks to a Wasserstein minimization method which permitted to 
estimating the best set of parameters of a given interaction kernel in the alignment opinion term. 
Possible future developments consist in improving the data-driven aspects of the model here presented by, for instance, replacing the parameter-dependent interaction kernel with its full reconstruction i.e. without assuming the apriori knowledge of its mathematical expression. A second direction that is worth to be explored regards the improvement of the model by assuming additional dependence on the opinion from the knowledge/education of the individuals. 

\section*{Acknowledgments}
This work has been written within the activities of the GNCS and GNFM groups of INdAM. This work was supported by the Italian Ministry of University and Research (MUR) through the PRIN 2020 project (No. 2020JLWP23) “Integrated Mathematical Approaches to Socio–Epidemiological Dynamics”, and by the MUR-PRIN Project 2022 PNRR No. P2022JC95T, “Data-driven discovery and control of multi-scale interacting artificial agent systems”, financed by the European Union - Next Generation EU.

\section{Appendix: particle-based kernel calibration}\label{sec:appendix}
In this section, we report the numerical procedure used in Section \ref{nume}, in particular for the calibration problem of the kernel parameters \eqref{eq:subkern}, introduced in sections \ref{test4} and \ref{test5}. To this end, we aim at minimizing the discrepancy between the opinion density obtained from numerical simulation and the score obtained from the sentiment analysis performed over data extracted from Twitter.

Thus we formulate the calibration of the kernel as a constrained optimization problem, which in the continuous form writes
\begin{align}
&\min_{\theta\in \Theta} \frac{1}{M}\sum_{m=1}^M\mathcal{D}_p(g(v,t_m|\theta),\hat g_m(\theta))\\
&\textrm{s.t.}\qquad
\mathcal{F}(f) = 0, \quad f^0(v,c)= f(v,c,0),\cr
&~\qquad g(v,t|\theta)= \int_{\R_+} f(v,c,t)\, dc,
\end{align}
where $ \mathcal{F}(f) = 0$ is a shorten notation for the Fokker-Planck equation \eqref{fpeq}, 
$g(v,t|\theta)$ represents the marginal opinion distribution of $f(v,c,t)$, $\hat g_m(v)$ represents the known distribution of opinions at time $\{t_m\}_{m=1}^M$, and $\mathcal{D}_p$ is a discrepancy measure, such as 
$p-$Wasserstein distance, or the $\ell_p$ distance,  with $p\geq 1$. In what follows  we propose, as a numerical approximation of this minimization procedure, a particle scheme based on two steps outlined in what follows.

\paragraph{Asymptotic particle-based scheme.} To simulate the evolution of the Fokker-Planck equation \eqref{fpeq} we rely on an asymptotic stochastic particle method to solve the kinetic dynamics \eqref{kine-ww} in the quasi-invariant regime \eqref{eq:scaling}. Then we introduce the following discretization
\begin{equation}\label{eq:dsmc}
f^{n+1} = \left(1-\frac{\Delta t}{\epsilon}\right)f^n + \frac{\Delta t}{\epsilon}Q^{\theta,+}_\epsilon(f^n,f^n),
\end{equation} 
where the gain operator $Q^{\theta,+}$ encodes the gain of particles in position $(v,c)$ at time $t$ after interactions \eqref{k1eps} and \eqref{eq.trules} have occurred.
The particle scheme for the simulation of \eqref{eq:dsmc} is reported in Algorithm \ref{nanbu}, where we set $\epsilon=\Delta t$ for simplicity. We refer to \cite{PT13} for details on this class of methods. 
Hence, this simulation scheme produces a sequence of data $\left\{(v^n_i,c^n_i)\right\}_{i=1}^{N_s}\sim f(v,c,t_n|\theta)$ for $n=0,\ldots,N_{t}$, that we can use in the next calibration procedure.
\begin{algorithm}
\caption{Asymptotic particle-based algorithm (Nanbu-like algorithm)}\label{nanbu}
\begin{algorithmic}[h!]
\State Fix $ 0<\epsilon = \Delta t<1$ and $N_s$.
\State Sample $\left\{v_i^0,c_i^0\right\}_{i=1}^{N_s}$ from the initial distribution $f_0(v,c)$
\For{$n = 0:N_t-1$}	
\State Set $N_c = round(N/2)$,
\State Select $N_c$ random pairs $(i, i_*)$ uniformly, without repetition among all possible pairs
\For {$i = 1:N_c$}. 
\State  Sample $\xi^n_i,\xi^n_{i_*}$ from a normal distribution $\mathcal N(0,1)$ and compute
\begin{equation*}
\begin{aligned}
	v_i^{n+1}         &=         v_i^{{n}}          + \epsilon \alpha P(v^n_i,v^n_{i_*},c^n_i,c^n_{i_*}|\theta) + \sqrt{\epsilon}\sigma D(v^n_i,c^n_i|\theta) \xi^n_i\cr
	v_{i_*}^{n+1}  &=          v_{i_*}^{{n}}  + \epsilon \alpha P(v^n_{i_*},v^n_i,c^n_{i_*},c^n_i|\theta)+ \sqrt{\epsilon}\sigma D(v^n_{i_*},c^n_{i_*}|\theta) \xi^n_{i_*}
\end{aligned}
\end{equation*}
\State Sample $\eta^n_{\epsilon,i}$ from a uniform random distribution s.t.  $\langle \eta^n_{\epsilon,i}\rangle =0$, $\langle (\eta^n_{\epsilon,i})^2\rangle =\epsilon\nu^2$,
\State  with values in the bound \eqref{eq:boundeta}.
\State Compute
\[
c_i^{n+1} =  c^{n}_i - \Psi^\epsilon_\delta(c^n_i/\bar c)c_i^n + \eta^n_{\epsilon,i} c_i^n.
\]
\EndFor
\EndFor
\end{algorithmic}
\end{algorithm}
\paragraph{Minimization of particle-based discrepancy} 
To minimize the discrepancy between the densities we rely directly on the information provided by the particles. Hence, from the particle simulation of  \eqref{eq:dsmc} we retrieve $ V^n(\theta)=\left\{v^n_i\right\}_{i=1}^{N_s}\sim g(v,t_n|\theta)$ for $n=1,\ldots,N_t$, whereas we recover the target particles sampling from the real-data distributions $N_s$ particles  $\hat V^m=\left\{\hat v^m_i\right\}_{i=1}^{N_s}\sim \hat g_m(v)$.
Hence we obtain the parameter $\theta^*$ as the minimizer of the following problem
\begin{equation}\label{eq:wass_1db}
\theta^*\in\arg\min_{\theta\in\Theta} \frac{1}{M}\sum_{m=1}^M\mathcal{D}_p(g_m^{N_s}(\theta),\hat g_m^{N_s}),
\end{equation} 
constrained to the evolution of the particle scheme \eqref{eq:dsmc}. Notice that 
the discrepancy $\mathcal{D}_p$is evaluated for the empirical densities $g_m^{N_s}(\theta),\hat g_m^{N_s}$ relative to the samples $V^m(\theta), \hat V^m$.


In order to perform the minimization of \eqref{eq:wass_1db}  we need to produce solutions in the admissible parameter space $\Theta$. Here we rely on \texttt{Matlab} routines for constrained minimization such as  \texttt{fmincon}, based on interior-point method, and \texttt{patternsearch}, as a gradient-free optimization method.
Indeed the fluctuations of the particle method are reflected in the discrepancy measure. Thus, to reduce the stochasticity induced by the particle simulation,  at each iteration of the optimization procedure, we have fixed the random-seed generator in Algorithm \ref{nanbu}.


\begin{remark}
Computing the discrepancy measure in \eqref{eq:wass_1db} can be challenging, for example in the aforementioned case of $p$-Wasserstein distance. However, in our case we can exploit the one-dimensional framework of the opinion space, hence computing equivalently 
\begin{equation}\label{eq:wass_1d}
\mathcal{D}_p( g_m^{N_s}(\theta),\hat g_m^{N_s})\equiv \mathcal{W}^p_p(g_m^{N_s}(\theta),\hat g_m^{N_s}) = \frac{1}{N_s}\sum_{i=1}^{Ns} |v^m_{\pi(i)}(\theta) - \hat v^m_{\pi'(i)}|^p,
\end{equation} 
where  $\pi$ and $\pi'$ are two permutation of the indices $1, \dots, Ns$ such that $v^m_{\pi(1)} \leq v^m_{\pi(2)}\leq \dots \leq v^m_{\pi(N_s)}$ and $\hat v^m_{\pi'(1)} \leq \hat v^m_{\pi'(2)}\leq \dots \leq\hat v^m_{\pi'(N_s)}$.
When $\ell_p$-distance is considered the discrepancy measure simply writes as follows
\begin{equation}\label{eq:wass_1dc}
\mathcal{D}_p( g_m^{N_s}(\theta),\hat g_m^{N_s})\equiv \int_{[-1,1]} |g_m^{N_s}(v|\theta)-\hat g_m^{N_s}(v)|^p\, dv,
\end{equation} 
where $g_m^{N_s}(\theta),\hat g_m^{N_s}$ have to be appropriately reconstructed  from the samples $V^m(\theta)$ and  $\hat V^m$.
\end{remark}

\begin{remark}
We remark that the optimization problem \eqref{eq:wass_1db} is in general a high-dimensional non-convex problem requiring efficient optimization methods see for example \cite{totzeck2021trends,borghi2023constrained,fornasier2021consensus}.
Different approaches are also advisable, reformulating the calibration into a function-approximation framework can give more generalizable results for the kernel inference, see for example \cite{bongini2017inferring,gottlich2022parameter,fiedler2023reproducing,chu2022inference}.
\end{remark}

\bibliographystyle{abbrv}
\bibliography{biblio_opinion}
\end{document}